\documentclass[12pt,peerreview,draftclsnofoot,letterpaper,oneside,onecolumn]{IEEEtran}
\usepackage{times}
\usepackage{amsmath}
\usepackage{float}
\usepackage{amsthm}
\usepackage{url}
\usepackage{setspace}
\usepackage{xcolor,graphicx,float}
\usepackage{subfig}
\usepackage{graphicx}
\usepackage{caption}
\usepackage{romannum}
\usepackage{cite}

\makeatletter
\def\thm@space@setup{\thm@preskip=0pt
\thm@postskip=0pt}
\makeatother
\newtheoremstyle{newstyle}      
{} 
{} 
{\mdseries} 
{} 
{\bfseries} 
{.} 
{ } 
{} 

\theoremstyle{newstyle}
\newtheorem{theorem}{Theorem}[section]
\newtheorem{corollary}{Corollary}[theorem]
\newtheorem{lemma}[theorem]{Lemma}
\newtheorem{proposition}[theorem]{Proposition}
\newtheorem{definition}[theorem]{Definition}
\newcommand{\E}{{\rm I\kern-.3em E}}

\begin{document}	
\pagenumbering{arabic}
%
\title{Regulating Competition in Age of Information under Network Externalities}
%
%
%

\author{Shugang~Hao,
        and~Lingjie~Duan,~\IEEEmembership{Senior Member,~IEEE}
\thanks{Part of this work has appeared in ACM MobiHoc 2019 Symposium \cite{hao2019economics}.}
\thanks{S. Hao and L. Duan are with the Pillar of Engineering Systems and Design, Singapore University of Technology and Design, Singapore, 487372 Singapore. E-mail: shugang\_hao@mymail.sutd.edu.sg, lingjie\_duan@sutd.edu.sg.}
}

\maketitle

\begin{abstract}
Online content platforms are concerned about the freshness of their content updates to their end customers, and increasingly more platforms now invite and pay the crowd to sample real-time information (e.g., traffic observations and sensor data) to help reduce their ages of information (AoI). How much crowdsourced data to sample and buy over time is a critical question for a platform's AoI management, requiring a good balance between its AoI and the incurred sampling cost. This question becomes more interesting by considering the stage after sampling, where multiple platforms coexist in sharing the content delivery network of limited bandwidth, and one platform's update may jam or preempt the others' under negative network externalities. When these selfish platforms know each other's sampling cost, we formulate their competition as a non-cooperative game and show they want to over-sample to reduce their own AoIs, causing the price of anarchy (PoA) to be infinity. To remedy this huge efficiency loss, we propose a trigger mechanism of non-monetary punishment in a repeated game to enforce the platforms' cooperation to approach the social optimum. We also study the more challenging scenario of incomplete information that some new platform hides its private sampling cost information from the other incumbent platforms in the Bayesian game. Perhaps surprisingly, we show that even the platform with more information may get hurt. We successfully redesign the trigger-and-punishment mechanism to negate the platform's information advantage and ensure no cheating. Our extensive simulations show that the mechanisms can remedy the huge efficiency loss due to platform competition, and the performance improves as we have more incumbent platforms with known cost information.   
\end{abstract}


\section{Introduction}

%
Today many customers do not want to lose any breaking news or useful information in smartphone even if in minute, and online platforms (such as social media outlets and navigation applications) want to keep their content updates fresh to attract a good number of customers for subscription and profit (\hspace{-0.5pt}\cite{NF,GM}). The platforms' updated real-time information can be news, traffic conditions, shopping promotions, restaurant discovery, and air quality conditions. 

Age of information (AoI) is a promising metric to characterize a platform's content update delay from an application layer point of view, and AoI measures the duration from the moment that the latest content was generated to the current reception time  \cite{kaul2012real}. Numerous works were done to analyze the AoI for a single link (\hspace{-0.5pt}\cite{kaul2012real, huang2015optimizing, costa2016age, sun2017update}). (\hspace{-0.5pt}\cite{bedewy2016optimizing, bedewy2017age, bedewy2019minimizing}) also analyzed the benefit of using queues to store outdated packets and improved the average age by choosing sampling rate. \cite{kaul2012status} extended the long-run AoI analysis to the case of multiple sources in a last-come-first-serve (LCFS) M/M/1 queue with given preemption policy. (\hspace{-0.5pt}\cite{hsu2017age,hsu2017scheduling}) optimized the online scheduling policy to balance multiple sources without knowing future data arrival patterns. 

The existing works on AoI focus on technological issues for controlling time-avarage age under different policies and scenarios (e.g., \cite{kaul2012real, huang2015optimizing, costa2016age, sun2017update, bedewy2016optimizing, bedewy2017age, bedewy2019minimizing, kaul2012status, hsu2017age,hsu2017scheduling}), and very few studies look at the economics of AoI management at the platform or system level. We are only aware that \cite{xuehe} studied how a single platform dynamically motivates sensors to sample fresh data on the source side, and \cite{meng} analyzed the purchase behavior to buy a platform's fresh data on the demand side. From the system management perspective, there are two critical issues to address. First, on the supply side, a platform needs to take care of the large sampling cost to support AoI. Increasingly more platforms now invite and pay the crowd to sample and send back real-time information (e.g., traffic observations, sales information, and sensor data) at large sampling rates \cite{duan2014motivating}. For example, online platform CrowdSpark follows this crowdsourcing approach and maintains a large pool of professional and citizen journalists who are paid to submit reports, news and videos. Another example is Waze platform who asks and rewards millions of drivers to report location-based observations (e.g., of road visibility, congestion, and {\textquotedblleft black-ice\textquotedblright} segments) when travelling in different routes of the city (\hspace{-0.5pt}\cite{li2017dynamic, jiang2017scalable}). We wonder how much crowdsourced data a platform should buy, expecting a balance between its AoI performance and the incurred sampling cost. 

The other economic issue is on the delivery side, where more than one selfish platform shares the same content delivery network for managing their individual AoIs. The content delivery network is of limited bandwidth and naturally involves competition among multiple platforms \; (\hspace{-0.5pt}\cite{tang2014regulating,lee2010max}). One platform's content update can jam or preempt the others' information updates, reducing its own AoI at the cost of the others'. How to enforce their cooperation despite the selfish nature of each is another key question, requiring new mechanism design under negative network externalities. In the literature, there are some game-theoretic studies on duopoly competition under externalities without mechanism design (\hspace{-0.5pt}\cite{duan2015economic,musacchio2006game}). \cite{varian2004system} further studied direct pricing or subsidy-based mechanisms to seek duopoly cooperation, yet such direct payment may be difficult to implement and realize in practice. For example, in our AoI management problem, it is difficult to ask the platforms to pay additionally according to their sample updates on top of their existing contracts with the Internet service provider (ISP). Regarding the literature of indirect (non-monetary) cooperation mechanism design for wireless networking applications, there are some repeated game studies that proposed trigger mechanisms of long-term punishment to hinder any platform's deviation from cooperation (e.g., \cite{duan2012attack,le2010repeated}). Yet these mechanisms are proposed for complete information or require sufficiently large discount factor when each platform cares enough for its future return. Differently, we design new trigger-and-punishment mechanisms here to work for any discount factor and incomplete information scenario. We also note that there are some pure economics studies on repeated games under incomplete information (e.g., \cite{aumann1995repeated}), yet they focus on signalling and learning and are not directly suitable for our problem of AoI management.

Our key novelty and main contributions are summarized as follows. 
\begin{itemize}
    \item \emph{Regulating AoI competition under network externalites}: To our best knowledge, this is the first paper studying the platform competition in AoI, and we take into account their sampling costs on the information supply side and update competition on the information delivery side. In Section \Romannum{2}, we model multiple selfish platforms' competition as non-cooperative games, depending on how well the platforms know about each other's sampling cost. 
    \item \emph{Huge efficiency loss due to platform competition:} Under complete information, each platform competes to increase its sampling rate without caring the others' AoI increases, and we prove the price of anarchy (PoA) is infinity. Under incomplete information where some platform newly joining the information market can hide its sample cost realization from the other incumbent platforms in the Bayesian game, we show surprisingly that the new platform may get hurt from gaining more information and the PoA is also infinity. 
    
    \item \emph{Trigger mechanism of non-monetary punishment for approaching the social optimum under complete information:} To remedy the huge efficiency loss under complete information, in Section \Romannum{3} we design a trigger mechanism of non-monetary punishment in the repeated game to enforce the platforms to cooperate, where we adapt the platforms' sampling cooperation profile for fitting any discount factor. As the discount factor increases, the mechanism's achieved performance improves to approach the social optimum as the platforms are more forward-looking to cooperate. 
    \item \emph{Approximate trigger-and-punishment mechanism  design under incomplete information:}
       To reduce the efficiency loss  between the Bayesian competition equilibrium and the social optimum under incomplete information, in Section \Romannum{4} we propose an approximate mechanism in the repeated game to enforce all the platforms' cooperation.
       Note that our mechanism under complete information does not work here, as the new platform now can take advantage of hiding its cost information to strategically  under- or over-sample without triggering the punishment. We successfully redesign the trigger-and-punishment mechanism to negate the platform's information advantage and ensure no cheating. We show the mechanism's performance improves as we have more incumbent platforms with known cost information.  
    \end{itemize}
\begin{figure}
\centering
\includegraphics[height=2.5in, width=4.2in]{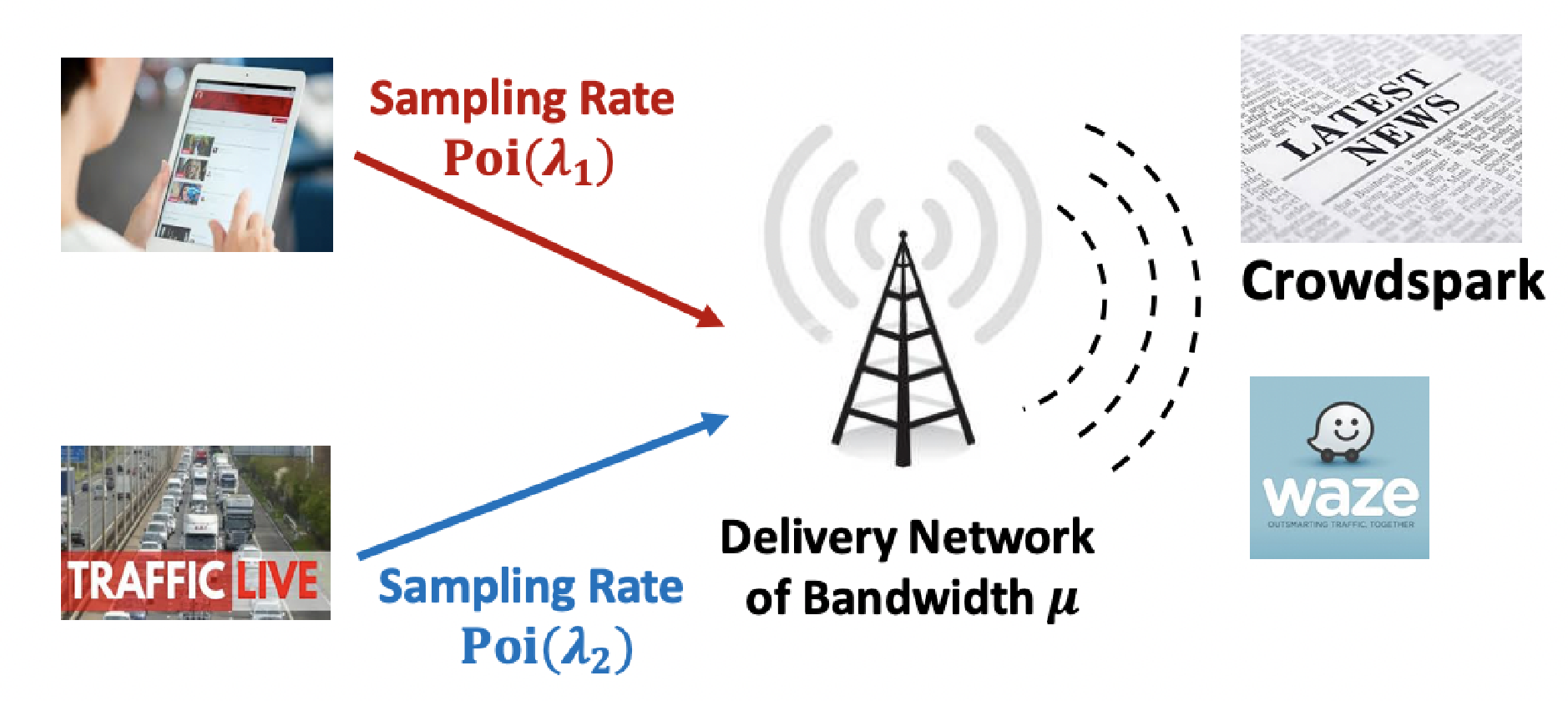}
\caption{Illustration of System Model. For example, here $N=2$ platforms (Crowdspark and Waze) respectively buy samples from their crowdsourcing pools at Poisson rates $\lambda_1$ and $\lambda_2$, and then update to their end customers through the same delivery network of bandwidth $\mu$. }
\end{figure}

\section{System Model and preliminary results}
As shown in Figure 1, $N$ online platforms (e.g., Crowdspark and Waze) first collect new samples (e.g., reports of spotted news and traffic observations) from their crowdsourcing pools at Poisson rates $\lambda_1$, $\lambda_2$, $\cdots$, $\lambda_N$, respectively, and then share the ISP's content delivery network of limited bandwidth $\mu$  to update new content to their customers in the same area. Here $\lambda_i$ denotes the mean rate of sampling generation of new information for platform $i\in\{1, 2, \cdots, N\}$. As in many of the AoI literature (e.g., \cite{kaul2012status, yates2018age1, yates2012real}), we assume in each platform's crowdsourcing pool, the sampling from each sensor source over time follows a Poisson process, and the total sampling to platform $i$ observations as superposition also follows Poisson process with mean rate $\lambda_i$. Platform $i$ can control mean rate $\lambda_i$ by providing proper incentive compensation to the crowdsourcing pool as in \cite{xuehe} and \cite{duan2014motivating}, and its average sampling cost is $c_i\lambda_i$ with unit compensation cost $c_i$.

 After sampling, we consider the content transmission time through the delivery network follows an exponential distribution with rate $\mu$. Without much loss of generality, we assume that the content delivery network applies a LCFS M/M/1 queue with preemption policy for processing the $N$ platforms' updates as in \cite{kaul2012status}, where the latest content arrival can preempt any platform's ongoing update in the network.\footnote{Though more involved, our model and the following analysis can also be extended to other queueing models such as a first-come-first-serve M/M/1 in \cite{kaul2012real}. There, the time-average AoI expression still shows the benefit of a platform to increase its sampling rate at the cost of increasing the other platforms' AoIs.} According to \cite{kaul2012status}, the time-average AoI at platform $i$ is given by 
 
 \begin{align}
    \Delta_i = \frac{\sum\limits_{j = 1}^N\lambda_j}{\lambda_i}(\frac{1}{\sum\limits_{j = 1}^N\lambda_j} + \frac{1}{\mu}), \nonumber
\end{align}
which decreases with its own sampling rate $\lambda_i$ and network bandwidth $\mu$, and increases with the other platforms' total sampling rate $\lambda_{-i} = \sum\limits_{j =1}^N\lambda_j - \lambda_i$ under negative network externalities.  

By further taking our modeled sampling cost into consideration, we define platform $i$'s total cost function as
\begin{align}
    \pi_i(\lambda_i, \lambda_{-i}) = \Delta_i + c_i\lambda_i, 
\end{align}
requiring platform $i$ to balance the AoI and the sampling cost when deciding its $\lambda_i$. Unlike each platform who only aims to minimize its own cost objective, the social planner wants to minimize the social cost as defined below: 
\begin{align}
     \pi(\lambda_1, \lambda_2, \cdots, \lambda_N) = \sum\limits_{i = 1}^N\pi_i(\lambda_i, \lambda_{-i}).\nonumber
\end{align}
In practice, a platform knows its own sampling cost yet may or may not know the other platforms' costs exactly. Next we present our preliminary results for the $N$ platforms' competition equilibrium under complete and incomplete information.
\subsection{Competition Equilibrium under Complete Information}
We first consider the information scenario that each cost $c_i$ is known to all the $N$ platforms. We formulate the $N$ platforms' interaction in the non-cooperative one-shot game where platform $i$ decides its own $\lambda_i$ to minimize its total cost $\pi_i(\lambda_i, \lambda_{-i})$ in (1) without considering the others' AoIs. As outcome of this game, we denote $(\lambda_1^*, \lambda_2^*, \cdots, \lambda_N^*)$ as the equilibrium sampling rates. 

To tell the maximum efficiency loss due to their selfish competition, we use the concept of price of anarchy (PoA) below:
\begin{align}
    PoA = \max_{c_1, c_2, \cdots, c_N, \mu} \frac{\pi(\lambda_1^*, \lambda_2^*, \cdots, \lambda_N^*)}{\pi(\lambda_1^{**}, \lambda_2^{**}, \cdots, \lambda_N^{**})}, \nonumber
\end{align}
where $(\lambda_1^{**}, \lambda_2^{**}, \cdots, \lambda_N^{**})$ denote the social optimizers when the $N$ platforms cooperate to jointly minimize the social cost. By checking the first-order conditions of convex costs $\pi(\lambda_1, \cdots, \lambda_N)$ and $\pi_i(\lambda_i, \lambda_{-i})$ with respect to $\lambda_i$ for all $i \in \{1, \cdots, N\}$, we have the following result. 
\begin{proposition}
Under complete information, the social optimizers $({\lambda_1^{**}}, {\lambda_2^{**}}, \cdots, \lambda_N^{**})$ are the unique solutions to
\begin{align}
    -\frac{1}{{\lambda}_i^2} (1+\frac{\sum\limits_{j = 1}^N \lambda_j - \lambda_i}{\mu}) + c_i + \frac{1}{\mu}\bigg(\sum\limits_{j = 1}^N\frac{1}{\lambda_j} - \frac{1}{\lambda_i}\bigg) = 0, \;\; \text{$i \in\{ 1,  \cdots, N\}$}.
\end{align}
Differently, the competition equilibrium $(\lambda_1^*, \lambda_2^*, \cdots, \lambda_N^*)$ are the unique solutions to
\begin{align}
    -\frac{1}{{\lambda}_i^2} (1+\frac{\sum\limits_{j = 1}^N \lambda_j - \lambda_i}{\mu}) + c_i = 0, \;\; \text{$i \in\{ 1,  \cdots, N\}$}.
\end{align}
 By comparing (2) and (3), we conclude that competition leads over-sampling  $(\lambda_i^* \geq \lambda_{i}^{**})$ for platform
 $i \in\{ 1, \cdots, N\}$ at the equilibrium.
\end{proposition}
The proof is given in Appendix A.
We notice from the third-term on the left-hand-side of (2) that at the social optimum platform $i$ cares about its sampling's negative externality effect $\frac{1}{\mu}\big({\sum\limits_{j = 1}^N}\frac{1}{\lambda_j}-\frac{1}{\lambda_i}\big)$ on the other $N-1$ platforms and will sample conservatively, while this is missing in (3) due to platform $i$'s selfishness at the equilibrium. 
The following result further explains their competition to over-sample at the equilibrium. 
\begin{corollary}
At the competition equilibrium, $\lambda_i^*$ increases with $\lambda_{j}^*$, and decreases with $c_i$, $c_j$ and $\mu$, respectively, where $i, j\in\{ 1, \cdots, N\}$ and $j \ne i$.
\end{corollary}
Intuitively, each platform worries that its update is preempted by the other platforms, and will sample and update more frequently. Such competition causes huge efficiency loss, as shown below. 
\begin{proposition}
 Price of anarchy under complete information is PoA=$\infty$, which is achieved when the smallest sampling cost among all the platforms $\min\{c_1, \cdots, c_N\}$ tends to be zero and the largest sampling cost among all the platforms $\max\{c_1, \cdots, c_N\}$ is non-trivial.
\end{proposition}
\begin{proof}
Suppose $c_i = \min\{c_1, \cdots, c_N\}$ and $c_j = \max\{c_1, \cdots, c_N\}$, we need to prove PoA=$\infty$ given $c_i \to 0$ and non-trivial $c_j > 0$.
As $c_i \to 0$, platform $i$ does not care its sampling cost and only aims to minimize its AoI. According to (3),  $\lambda_i^*$ will go to infinity and this stimulates platform $j$'s $\lambda_j^* = \sqrt{\frac{1+(\sum\limits^N_{k =1}\lambda_k^*-\lambda_j^*)/\mu}{c_j}}$ to reach infinity. However, at the social optimum, both $(\lambda_1^{**}, \lambda_2^{**}, \cdots, \lambda_N^{**})$ in (2) and the resultant social cost $\pi(\lambda_1^{**}, \lambda_2^{**}, \cdots, \lambda_N^{**})$ are finite. Then we have
    \begin{align}
    PoA &= \lim_{c_i \to 0} \frac{\sum\limits_{k=1}^N\frac{\sum\limits_{l = 1}^N\lambda_l^*}{\lambda_k^*}\bigg(\frac{1}{\sum\limits_{l = 1}^N\lambda_l^*} + \frac{1}{\mu}\bigg)+ \sum\limits_{l = 1}^{N}c_l\lambda_l^*}{\pi(\lambda_1^{**}, \lambda_2^{**}, \cdots, \lambda_N^{**})} \geq \lim_{c_i \to 0} \frac{c_j\lambda_j^*}{\pi(\lambda_1^{**}, \lambda_2^{**}, \cdots, \lambda_N^{**})} = \infty. \nonumber  \qedhere
\end{align}
\end{proof}
This huge efficiency loss motivates us to design non-monetary cooperation mechanisms to mitigate the competition among the platforms. Note that the social optimum can be easily realized if the social planner (e.g., the ISP) can charge a monetary penalty $\frac{1}{\mu}\big(\sum\limits_{j = 1}^N\frac{1}{\lambda_j^{**}}-\frac{1}{\lambda_i^{**}}\big)$ per sampling rate $\lambda_i$ from platform $i$, where $i \in\{1, \cdots, N\}$. However, this additional charging based on usage is intrusive and difficult to implement in practice, given the content platforms' existing flat contracts with the ISP. This motivates us to design non-monetary cooperation mechanisms in the repeated game in Section \Romannum{3}, which is more challenging. 
\subsection{Competition Equilibrium under Incomplete Information}
Now we consider the incomplete information scenario that there is 1 newly joined platform in the network (namely, platform 1) whose unit sampling cost is time-varying and its realization is only known to itself, while the other $N-1$ existing platforms know each other's sampling cost exactly and are uncertain about that of the new platform.\footnote{We can also extend our analysis to another incomplete scenario where more than one platform can hide its cost information from the others, though the analysis involving many combinations of cost realizations becomes more complicated.} Accordingly, we model the public information that platform 1 has probability $p_H$ of having high sampling cost $c_H$ and probability $1-p_H$ of having low sampling cost $c_L$ each time it samples, while
 sampling cost $c_i$ of platform $i\in\{2, \cdots, N\}$ is constant over time. On one hand, platform 1 knows its $c_1$ realization ($c_L$ or $c_H$) and the other $c_i$'s. On the other hand, platform $i$ only knows all costs $c_i$'s and the probability distribution of $c_1$, and it is also aware of platform 1's information advantage. In other words, the ratio $1/N$ tells the degree of incomplete information to the social planner or public. We wonder if platform 1 benefits from this and if we have huge efficiency loss result as in complete information. To let platform 1 fully uses its information advantage, we consider a more challenging case where unit sampling costs of platforms follow the order of $p_Hc_H + (1-p_H)c_L \leq c_2 \leq \cdots \leq c_N$ with smallest mean cost for platform 1. 

Since platform 1 knows its own unit cost exactly, it will take this information advantage and adaptively decide $\lambda_1(c_H)$ when $c_1 = c_H$ and $\lambda_1(c_L)$ when $c_1 = c_L$. Unaware of $c_1$ realizations, platform $i\in\{2, \cdots, N\}$ behaves indifferently to decide $\lambda_i$ constantly over time. We model such platform competiton as a Bayesian game as follows.  

Given $c_1=c_H$, platform 1's cost function is
\begin{align}
    \pi_{1}\left(\lambda_{1}\left(c_{H}\right), \sum\limits_{i=2}^{N} \lambda_{i}\right)=\frac{\lambda_{1}\left(c_{H}\right)+\sum\limits_{i=2}^{N} \lambda_{i}}{\lambda_{1}\left(c_{H}\right)}\left(\frac{1}{\lambda_{1}\left(c_{H}\right)+\sum\limits_{i=2}^{N} \lambda_{i}}+\frac{1}{\mu}\right)+c_{H} \lambda_{1}\left(c_{H}\right),
\end{align}
and otherwise,
\begin{align}
    \pi_{1}\left(\lambda_{1}\left(c_{L}\right), \sum\limits_{i=2}^{N} \lambda_{i}\right)=\frac{\lambda_{1}\left(c_{L}\right)+\sum\limits_{i=2}^{N} \lambda_{i}}{\lambda_{1}\left(c_{L}\right)}\left(\frac{1}{\lambda_{1}\left(c_{L}\right)+\sum\limits_{i=2}^{N} \lambda_{i}}+\frac{1}{\mu}\right)+c_{L} \lambda_{1}\left(c_{L}\right).
\end{align}
Under incomplete information, the cost function of platform $i\in\{2, \cdots, N\}$ is defined below in average sense:
\begin{align}
    &\pi_{i}\left(\left(\lambda_{1}\left(c_{H}\right), \lambda_{1}\left(c_{L}\right)\right), \lambda_{i}, \lambda_{-i}\right)=p_{H} \cdot\left(\frac{\lambda_{1}\left(c_{H}\right)+\lambda_{i}+\lambda_{-i}}{\lambda_{i}}\left(\frac{1}{\lambda_{1}\left(c_{H}\right)+\lambda_{i}+\lambda_{-i}}+\frac{1}{\mu}\right)\right)  \nonumber \\
    &+\left(1-p_{H}\right) \cdot\left(\frac{\lambda_{1}\left(c_{L}\right)+\lambda_{i}+\lambda_{-i}}{\lambda_{i}}\left(\frac{1}{\lambda_{1}\left(c_{L}\right)+\lambda_{i}+\lambda_{-i}}+\frac{1}{\mu}\right)\right)+c_{i} \lambda_{i},
\end{align}
where $\lambda_{-i} = \sum\limits_{j = 2}^N\lambda_j - \lambda_i$.
The average social cost function is
\begin{align}
    \pi\left(\left(\lambda_{1}\left(c_{H}\right), \lambda_{1}\left(c_{L}\right)\right), \lambda_2, \cdots, \lambda_N\right)=&p_{H} \pi_{1}\left(\lambda_{1}\left(c_{H}\right), \sum\limits_{i=2}^{N} \lambda_{i}\right) +\left(1-p_{H}\right) \pi_{1}\left(\lambda_{1}\left(c_{L}\right), \sum\limits_{i=2}^{N} \lambda_{i}\right)\nonumber \\
    &+\sum\limits_{i=2}^{N} \pi_{i}\left(\left(\lambda_{1}\left(c_{H}\right), \lambda_{1}\left(c_{L}\right)\right), \lambda_{i}, \lambda_{-i}\right).
\end{align}
Similar to the complete information scenario, we present the concept of PoA below in this incomplete information scenario:
\begin{align}
    P_{o} A=\max _{c_{L}, c_{H}, c_{2}, \cdots, c_{N}, \mu, p_{H}} \frac{\pi\left(\left(\lambda_{1}^{*}\left(c_{H}\right), \lambda_{1}^{*}\left(c_{L}\right)\right), \lambda_{2}^{*}, \cdots, \lambda_{N}^{*}\right)}{\pi\left(\left(\lambda_{1}^{* *}\left(c_{H}\right), \lambda_{1}^{* *}\left(c_{L}\right)\right), \lambda_{2}^{* *}, \cdots, \lambda_{N}^{* *}\right)}, \nonumber
\end{align}
 where $\left(\left(\lambda_{1}^{*}\left(c_{H}\right), \lambda_{1}^{*}\left(c_{L}\right)\right), \lambda_{2}^{*}, \cdots, \lambda_{N}^{*}\right)$ summarize competition equilibrium and social optimizers are summarized as $\left(\left(\lambda_{1}^{* *}\left(c_{H}\right), \lambda_{1}^{* *}\left(c_{L}\right)\right), \lambda_{2}^{* *}, \cdots, \lambda_{N}^{* *}\right)$. By checking the first-order conditions of convex costs $\pi_{1}\left(\lambda_{1}\left(c_{H}\right), \sum\limits_{i=2}^{N} \lambda_{i}\right)$ in (4), $\pi_{1}\left(\lambda_{1}\left(c_{L}\right), \sum\limits_{i=2}^{N} \lambda_{i}\right)$ in (5), $\pi_{i}(\left(\lambda_{1}\left(c_{H}\right), \lambda_{1}\left(c_{L}\right)\right), \lambda_{i}$, $\lambda_{-i})$ in (6) and $\pi\left(\left(\lambda_{1}\left(c_{H}\right), \lambda_{1}\left(c_{L}\right)\right), \lambda_{2}, \cdots, \lambda_{N}\right)$ in (7) with respect to $\lambda_1(c_H)$, $\lambda_1(c_L)$ and $\lambda_i$, $i\in\{2, \cdots, N\}$, respectively, we have the following result. 
\begin{proposition}
Under incomplete information, $\left(\left(\lambda_{1}^{* *}\left(c_{H}\right), \lambda_{1}^{* *}\left(c_{L}\right)\right), \lambda_{2}^{* *}, \cdots, \lambda_{N}^{* *}\right)$ as the social optimizers are the unique solutions to
\vspace{-5pt}
\begin{table}[H]
    \small{\begin{align}
    &-\frac{1}{\lambda_{1}^{2}\left(c_{H}\right)}\left(1+\frac{\sum\limits_{i=2}^{N} \lambda_{i}}{\mu}\right)+c_{H}+\frac{1}{\mu} \sum\limits_{i=2}^{N} \frac{1}{\lambda_{i}}=0,  \\
    &-\frac{1}{\lambda_{1}^{2}\left(c_{L}\right)}\left(1+\frac{\sum\limits_{i=2}^{N} \lambda_{i}}{\mu}\right)+c_{L}+\frac{1}{\mu} \sum\limits_{i=2}^{N} \frac{1}{\lambda_{i}}=0,  \\
    &p_{H}\left(-\frac{1}{\lambda_{i}^{2}}\left(1+\frac{\lambda_{1}\left(c_{H}\right)+\sum\limits_{j =2}^{N} \lambda_{j}-\lambda_i}{\mu}\right)+c_{i}+\frac{1}{\lambda_{1}\left(c_{H}\right) \mu}\right) \nonumber \\
    &+\left(1-p_{H}\right)\left(-\frac{1}{\lambda_{i}^{2}}\left(1+\frac{\lambda_{1}\left(c_{L}\right)+\sum\limits_{j =2}^{N} \lambda_{j}-\lambda_i}{\mu}\right)+c_{i}+\frac{1}{\lambda_{1}\left(c_{L}\right) \mu}\right)+\frac{1}{\mu}\bigg( \sum\limits_{j =2}^{N} \frac{1}{\lambda_{j}}-\frac{1}{\lambda_i}\bigg)=0,
\end{align}
    }
    \vspace{-20pt}
\end{table}
\noindent where $i\in\{2, \cdots, N\}$. Differently, the competition equilibrium
$\left(\left(\lambda_{1}^{*}\left(c_{H}\right), \lambda_{1}^{*}\left(c_{L}\right)\right), \lambda_{2}^{*}, \cdots, \lambda_{N}^{*}\right)$ are the unique solutions to
\vspace{-5pt}
\begin{table}[H]
   \small{
   \begin{align}
    &-\frac{1}{\lambda_{1}^{2}\left(c_{H}\right)}\left(1+\frac{\sum\limits_{i=2}^{N} \lambda_{i}}{\mu}\right)+c_{H}=0,   \\
     &-\frac{1}{\lambda_{1}^{2}\left(c_{L}\right)}\left(1+\frac{\sum\limits_{i=2}^{N} \lambda_{i}}{\mu}\right)+c_{L}=0,   \\
    &-\frac{p_{H}}{\lambda_{i}^{2}}\left(1+\frac{\lambda_{1}\left(c_{H}\right)+\sum\limits_{j =2}^{N} \lambda_{j}-\lambda_i}{\mu}\right)-\frac{1-p_{H}}{\lambda_{i}^{2}}\left(1+\frac{\lambda_{1}\left(c_{L}\right)+\sum\limits_{j =2}^{N} \lambda_{j}-\lambda_i}{\mu}\right)+c_{i}=0,
\end{align}
   }
   \vspace{-20pt}
\end{table}
\noindent where $i\in\{2, \cdots, N\}$.
All the platforms will over-sample at equilibrium, i.e., $\lambda_1^*(c_H)\geq \lambda_1^{**}(c_H) $, $\lambda_1^*(c_L)\geq \lambda_1^{**}(c_L) $ and $\lambda_i^* \geq \lambda_i^{**}$. Additionally, we have $\lambda_1^*(c_H)/\lambda_1^*(c_L) = \sqrt{c_L/c_H}$.
\end{proposition}
The proof is given in Appendix B.
Similar to the complete information scenario, all the platforms unnecessarily over-sample at competition equilibrium. Such competition also causes huge efficiency loss, as shown below by following a similar proof of Proposition \Romannum{2}.2. 
\begin{proposition}
Price of anarchy under incomplete information is PoA$=\infty$, which is achieved when the smaller cost $c_L$ of platform 1 tends to be zero while some $c_i$  with $i\in\{2,...,N\}$ is non-trivial. 
\end{proposition}
\begin{proof}
We want to show PoA$=\infty$ given $c_L \to 0$ and non-trivial some $c_i > 0$ for platform $i\in\{2,...,N\}$.
As $c_L \to 0$, platform 1 does not care its sampling cost when $c_1 = c_L$ and only aims to minimize its AoI, which is decreasing in $\lambda_1(c_L)$ according to (5). Thus optimal $\lambda_1^*(c_L)$ will go to infinity.

As $c_i > 0$ is non-trivial, $\lambda_1^*(c_L)$ going to infinity stimulates platform $i$'s sampling rate $\lambda_i^*$ as $\lambda_i^* = \sqrt{\frac{1+(p_H\lambda_1^*(c_H)+(1-p_H)\lambda_1^*(c_L)+\sum\limits_{j = 2}^N\lambda_j^*-\lambda_i^*)/\mu}{c_i}}$ to reach infinity. However, at the social optimum, both $\left(\left(\lambda_{1}^{* *}\left(c_{H}\right), \lambda_{1}^{* *}\left(c_{L}\right)\right), \lambda_{2}^{* *}, \cdots, \lambda_{N}^{* *}\right)$ in (8)-(10) and the resultant social cost $\pi((\lambda_{1}^{* *}(c_{H})$, $\lambda_{1}^{* *}(c_{L}))$, $ \lambda_{2}^{* *}, \cdots, \lambda_{N}^{* *})$ are finite. Then we have
\begin{align}
    PoA &=  \lim_{c_L \to 0} \frac{\frac{p_H}{\lambda_1^{*}(c_H)}(1+{\sum\limits_{j=2}^N\lambda_j^{*}}/\mu) + \frac{1-p_H}{\lambda_1^{*}(c_L)}(1+{\sum\limits_{j=2}^N\lambda_j^{*}}/\mu)+ (1-p_H)c_L\lambda_1^*(c_L) +{\sum\limits_{j =2}^N c_j\lambda_j^{*}} }{\pi((\lambda_1^{**}(c_H), \lambda_1^{**}(c_L)), \lambda_2^{**}, \cdots, \lambda_N^{**})}\nonumber \\
    &+\frac{\sum\limits_{j=2}^N\frac{1}{\lambda_j^{*}}
    \big(1+\big(p_H\lambda_1^*(c_H) + (1-p_H)\lambda_1^*(c_L) + {\sum\limits_{k =2}^N\lambda_k^{*}} - \lambda_j^*\big)/\mu\big) + p_H c_H\lambda_1^*(c_H) + {N}/\mu}{\pi((\lambda_1^{**}(c_H), \lambda_1^{**}(c_L)), \lambda_2^{**}, \cdots, \lambda_N^{**})} \nonumber \\
    &\geq \frac{ c_i\lambda_i^{*}}{\pi((\lambda_1^{**}(c_H), \lambda_1^{**}(c_L)), \lambda_2^{**}, \cdots, \lambda_N^{**})} = \infty.  \nonumber \qedhere
\end{align}
\end{proof}
If the social planner or ISP can charge $\frac{1}{\mu} \sum\limits_{i=2}^{N} \frac{1}{\lambda_{i}^{* *}}$ per update from platform 1 and $\frac{p_{H}}{\lambda_{1}^{**}\left(c_{H}\right) \mu}+\frac{1-p_{H}}{\lambda_{1}^{**}\left(c_{L}\right) \mu}+\frac{1}{\mu}\big( \sum\limits_{j = 2}^N \frac{1}{\lambda_{j}^{* *}}- \frac{1}{\lambda_i^{**}}\big)$ from platform $i\in\{2, \cdots, N\}$, then the social optimum can be achieved. 
However, this additional charging based on usage is intrusive and difficult to implement in practice, given the content platforms' existing flat contracts with the ISP. This motivates us to design non-monetary cooperation mechanisms in the repeated game in Section \Romannum{4}, which is more challenging. 

Finally, we check platform 1's equilibrium cost objective and wonder if it takes advantage from knowing more information of its own cost realization.
\begin{proposition}
Under incomplete information, the one-shot cost of platform 1 when $c_1=c_H$ is greater than that under complete information, and once $p_H$ is large, even its time-average cost $p_H\pi_1(\lambda_1^*(c_H), \lambda_2^*, \cdots, \lambda_N^* )+(1-p_H)\pi_1(\lambda_1^*(c_L), \lambda_2^*, \cdots, \lambda_N^*)$ becomes greater than that under complete information.
\end{proposition}
The proof is given in Appendix C.
Under complete information, when $c_1=c_H$ platform 1 does not want to sample much to save its sampling cost, and platform $i\in\{2, \cdots, N\}$ knowing $c_1=c_H$ expects weak limited negative network externalities from platform 1 and also samples conservatively. However, under incomplete information, platform $i$ can no longer observe $c_1=c_H$ or $c_1=c_L$ instances,  and its over-sampling when $c_1=c_H$ forces platform 1 to over-sample, intensifying the competition and hurting all. Once $p_H$ is large, this happens more often and even platform 1 loses on average sense.

\section{Trigger-and-punishment Mechanism under Complete Information}
To remedy the huge inefficiency with PoA$=\infty$ proved in Proposition \Romannum{2}.2, we want to stimulate cooperation between the $N$ platforms. Without direct pricing or penalty, this is difficult to enforce in one-shot, and thus we propose to use an infinitely repeated game to shift the $N$ platforms' myopic decision-making to be more forward-looking in the long run. In this repeated game, all the platforms will simultaneously  play the non-cooperative one-shot game in Section \Romannum{2}.A for infinitely many rounds with discount factor $\delta \in (0, 1)$. Note that $\delta$ tells how much a platform evaluates its one-shot cost in next round as compared to the current cost. 
Yet this repeated game alone is not enough to ensure cooperation, as each platform will still behave the same as $(\lambda_1^*,\lambda_2^*, \cdots, \lambda_N^*)$ in (3) in each round. We next propose a trigger mechanism of indirect punishment as credible threat to prevent their myopic over-sampling in the first place. Note that a platform's AoI decreases with its update and increases with the other platforms' updates under negative network externalities.
\begin{definition}
\emph{Our non-forgiving trigger mechanism of indirect punishment under complete information works in the following:}
 \begin{itemize}
    \item \emph{In each round, each platform follows cooperation profile $\big(\tilde{\lambda}_1(\delta)$, $\tilde{\lambda}_2(\delta)$,$\cdots$, $\tilde{\lambda}_N(\delta)\big)$ to sample if none was ever detected to deviate from this profile in the past.}\footnote{Here, each time slot in the repeated game is long enough for each platform's AoI statistic to converge to its average value. Then platform $i$ can easily identify the other platforms' total sampling rate $\lambda_{-i}$ from its own average AoI experience $\Delta_i(\lambda_i, \lambda_{-i})$ in (1) and rate $\lambda_i$. Note that each platform only has intention to over-sample. As long as one platform really over-samples, $\lambda_{-i}$ increases and all the other platforms can infer deviation to trigger the punishment. }
    \item \emph{Once a deviation was found in the past, the $N$ platforms will keep playing the equilibrium profile $(\lambda_1^*, \lambda_2^*, \cdots, \lambda_N^*)$} in (3) forever as punishment.
\end{itemize}
\end{definition}
We expect the social planner (e.g., the ISP) to implement the cooperation mechanisms and
recommend the cooperation or punishment profile to platforms based on their operations overtime.
Our mechanism as described above has another advantage: to trigger the punishment, we do not need to identify which platform deviates. One can imagine that as long as a platform cares enough for future costs  under a large discount factor $\delta$, it is unlikely to deviate to trigger severe punishment. It should be noted that in the extreme case of $\delta\rightarrow 0$, each platform only cares for immediate cost and $\big(\tilde{\lambda}_1(\delta)$, $\tilde{\lambda}_2(\delta)$,$\cdots$, $\tilde{\lambda}_N(\delta)\big)$ degenerate to $(\lambda_1^*, \lambda_2^*, \cdots, \lambda_N^*)$ in the one-shot game.  We next design the cooperation profile $\big(\tilde{\lambda}_1(\delta)$, $\tilde{\lambda}_2(\delta)$,$\cdots$, $\tilde{\lambda}_N(\delta)\big)$ according to any value of non-trivial $\delta$. 
\subsection{Cooperation Profile Design for Large $\delta$ Regime}
In this subsection, we first suppose the social optimum is attainable via our repeated game with $\big(\tilde{\lambda}_1(\delta)$, $\tilde{\lambda}_2(\delta)$,$\cdots$, $\tilde{\lambda}_N(\delta)\big)$ = $(\lambda_1^{**}, \lambda_2^{**}, \cdots, \lambda_N^{**})$ in (2), then any platform's deviation  will bring itself in a larger long-term cost. We can use this no-deviation condition to reverse-engineer the feasible regime of $\delta$ for enabling such $(\lambda_1^{**}, \lambda_2^{**}, \cdots, \lambda_N^{**})$ in the first place.

If platform $i \in \{1, \cdots, N\}$ chooses to deviate to any $\lambda_i$, it is optimal to deviate in the first round to save the immediate cost $\pi_i(\lambda_i, \lambda_{-i}^{**})$ in (1) without any time discount. Its optimal deviation or best response to $\lambda_{-i}^{**}$ is $\lambda_i=\sqrt{\frac{1+\lambda_{-i}^{**}/\mu}{c_i}}$ according to (3). Its (discounted) long-term cost objective over all time stages is 
\begin{align}
    \hat{\Pi}_i  &=  \pi_i\bigg(\sqrt{\frac{1+\lambda_{-i}^{**}/\mu}{c_i}}, \lambda_{-i}^{**}\bigg) + \delta \pi_i(\lambda_i^*, \lambda_{-i}^*) + \delta^2\pi_i(\lambda_i^*, \lambda_{-i}^*) + \cdots, \nonumber \\
    &= \pi_i\bigg(\sqrt{\frac{1+\lambda_{-i}^{**}/\mu}{c_i}}, \lambda_{-i}^{**}\bigg) + \frac{\delta}{1-\delta} \pi_i(\lambda_i^*, \lambda_{-i}^*),
\end{align}
\noindent where punishment is triggered since time stage 2 and $\lambda_{-i}^{*}=\sum\limits_{j=1}^N\lambda_j^{*} - \lambda_i^*$. 
Otherwise, it will always cooperate and obtain the following cost without any deviation,
\begin{align}
    \Pi_i = \pi_i(\lambda_i^{**}, \lambda_{-i}^{**}) + \delta \pi_i(\lambda_i^{**}, \lambda_{-i}^{**}) + \delta^2 \pi_i(\lambda_i^{**}, \lambda_{-i}^{**}) + \cdots = \frac{1}{1-\delta}\pi_i(\lambda_i^{**}, \lambda_{-i}^{**}). 
\end{align}
To ensure that platform {$i$} never deviates in the repeated game, we require {$\hat{\Pi}_i\geq \Pi_i$}, or simply 
\begin{align}
\delta \geq \delta^{th}_i := \frac{\pi_i(\lambda_i^{**}, \lambda_{-i}^{**}) - \pi_i\bigg(\sqrt{\frac{1+\lambda_{-i}^{**}/\mu}{c_i}}, \lambda_{-i}^{**}\bigg)}{\pi_i(\lambda_i^*, \lambda_{-i}^*) - \pi_i\bigg(\sqrt{\frac{1+\lambda_{-i}^{**}/\mu}{c_i}}, \lambda_{-i}^{**}\bigg)} =\frac{\bigg(\sqrt{\frac{1+\lambda_{-i}^{**}/\mu}{c_i+\frac{1}{\mu}\big(\sum\limits_{j = 1}^N\frac{1}{\lambda_j^{**}}-\frac{1}{\lambda_i^{**}}\big)}} - \sqrt{\frac{1+\lambda_{-i}^{**}/\mu}{c_i}}\bigg)^2}{2\lambda_i^{**}\bigg(\lambda_i^*-\sqrt{\frac{1+\lambda_{-i}^{**}/\mu}{c_i}}\bigg)}. 
\end{align} 
Without loss of generality, we assume {$c_1 \leq c_2 \leq \cdots \leq c_N$} and can show that platform 1 is more likely to deviate with {$\delta^{th}_1\geq \delta^{th}_2 \geq \cdots \geq \delta^{th}_N$}. The following summarizes the trigger mechanism with perfect cooperation profile for {$\delta\geq \max(\delta^{th}_1, \delta^{th}_2, \cdots, \delta^{th}_N)=\delta^{th}_1$}. 

\begin{proposition}[Large $\delta$ Regime]
Under complete information, if $\delta \geq \delta^{th}_1$ with $\delta_{i=1}^{th}$ in (16), all the platforms will follow the perfect cooperation profile {$(\tilde{\lambda}_1(\delta), \tilde{\lambda}_2(\delta), \cdots, \tilde{\lambda}_N(\delta))=(\lambda_1^{**}, \lambda_2^{**}, \cdots, \lambda_N^{**})$} in (2) all the time, without triggering the punishment profile {$(\lambda_1^*,\lambda_2^*, \cdots, \lambda_N^*)$} in (3).

\end{proposition}
 The threshold $\delta^{th}_i$ tells platform $i$'s unwillingness to cooperate and we prefer small threshold for this platform to follow $\lambda_i^{**}$ ideally. 

 


\subsection{Cooperation Profile Design for Medium $\delta$ Regime}

If {$\delta^{th}_{j+1} \leq \delta < \delta^{th}_j$, where $j \in \{1, \cdots, N-1\}$}, platform {$k \in \{j+1, \cdots, N\}$} will still follow the social optimizer {$\lambda_k^{**}$} yet platform {$i \in \{1, \cdots, j\}$} with smaller costs will deviate, requiring us to design new {$\tilde{\lambda}_i(\delta)$} as a function of $\delta$ to replace {$\lambda_i^{**}$ for such platforms}. By ensuring the long-term cost {$\Pi_i(\tilde{\lambda}_i(\delta),\lambda_{-i}(\delta))$} in (15) without deviation just equal to (14) with the best deviation, where $\lambda_{-i}(\delta) = \sum\limits_{k =1 }^{j}\tilde{\lambda}_k(\delta) - \tilde{\lambda}_i(\delta) + \sum\limits_{k = j+1}^N\lambda_k^{**}$, we optimally determine the {$\tilde{\lambda}_i(\delta)$} by jointly solving the following equations 
{
\begin{align}
\delta = \frac{\pi_i(\tilde{\lambda}_i(\delta), \lambda_{-i}(\delta)) - \pi_i\bigg(\sqrt{\frac{1+\lambda_{-i}(\delta)/\mu}{c_i}}, \lambda_{-i}(\delta)\bigg)}{\pi_i(\lambda_i^*, \sum\limits_{j = 1}^N\lambda_j^*-\lambda_i^*) -\pi_i\bigg(\sqrt{\frac{1+\lambda_{-i}(\delta)/\mu}{c_i}}, \lambda_{-i}(\delta)\bigg)}, \;\; i \in\{ 1,\cdots, j\}.
\end{align}
}
When solving (17), there are two candidates for $\tilde{\lambda}_i(\delta)$ and we choose to take the smaller root with smaller social cost. Then we have the following result. 

\begin{proposition}[Medium $\delta$ Regime]
In the repeated game under complete information, if $\delta^{th}_{j+1} \leq \delta < \delta^{th}_j$, where $j \in \{1, \cdots, N-1\}$, all the platforms will always follow the cooperation profile below without deviating to trigger punishment $(\lambda_1^*,\lambda_2^*, \cdots, \lambda_N^*)$ in (3): 
\begin{itemize}
    \item For platform $k \in \{j+1, \cdots, N\}$ with larger unit sampling costs: $\tilde{\lambda}_k(\delta) = \lambda_k^{**}$.
    \item For platform $i \in \{1, \cdots, j\}$: ($\tilde{\lambda}_1(\delta), \cdots, \tilde{\lambda}_j(\delta)$) are unique solutions to
    \begin{align}          
    \delta\lambda_i^* + (1-\delta)\sqrt{\frac{1+\frac{\lambda_{-i}(\delta)}{\mu}}{c_i}} 
   -\sqrt{\bigg(\delta\lambda_i^* + (1-\delta)\sqrt{\frac{1+\frac{\lambda_{-i}(\delta)}{\mu}}{c_i}} \bigg)^2 - \frac{1+\frac{\lambda_{-i}(\delta)}{\mu}}{c_i}} - \tilde{\lambda}_i(\delta) = 0. \nonumber
        \end{align}
    Here, $\lambda_i^{**} < \tilde{\lambda}_i(\delta) < \lambda_i^{*}$ and $\tilde{\lambda}_i(\delta)$ decreases with $\delta$.
\end{itemize}

\end{proposition}

The proof is given in Appendix D. As we prefer the platforms not to over-sample, a larger $\delta$ in the medium regime helps.

\subsection{Cooperation Profile Design for Small $\delta$ Regime}

If $\delta$ is smaller than the smallest threshold $\delta^{th}_N$, no platform will follow the social optimizers, and we need to design totally new {$\big(\tilde{\lambda}_1(\delta), \tilde{\lambda}_2(\delta), \cdots, \tilde{\lambda}_N(\delta)\big)$} as functions of $\delta$ jointly. Similar to (17), we now have   
\begin{align}
   \delta = \frac{\pi_i(\tilde{\lambda}_{i}(\delta), \tilde{\lambda}_{-i}(\delta)) - \pi_i\bigg(\sqrt{\frac{1+\tilde{\lambda}_{-i}(\delta)/\mu}{c_i}}, \tilde{\lambda}_{-i}(\delta)\bigg)}{\pi_i(\lambda_i^*, \lambda_{-i}^*) - \pi_i\bigg(\sqrt{\frac{1+\tilde{\lambda}_{-i}(\delta)/\mu}{c_i}}, \tilde{\lambda}_{-i}(\delta)\bigg)}, \;\; i \in\{ 1, \cdots, N\},  
\end{align}
where $\tilde{\lambda}_{-i}(\delta) = \sum\limits_{j =1}^N \tilde{\lambda}_{j}(\delta) - \tilde{\lambda}_{i}(\delta)$ and $\lambda_{-i}^* = \sum\limits_{j=1}^N\lambda_j^*-\lambda_i^*$. After solving (18) and taking the smaller roots for all $\tilde{\lambda}_i(\delta)$'s to avoid large social cost, we have the following result. 

\begin{proposition}[Small $\delta$ Regime]
 In the repeated game under complete information, if  $\delta < \delta^{th}_N$ with $i = N$ in (16), all the platforms will always follow the cooperation profile $\big(\tilde{\lambda}_1(\delta), \tilde{\lambda}_2(\delta), \cdots, \tilde{\lambda}_N(\delta) \big)$ below as unique solutions to 
 \begin{align}
   \tilde{\lambda}_i(\delta) - \delta\lambda_i^* - (1-\delta)\sqrt{\frac{1+\frac{\tilde{\lambda}_{-i}(\delta)}{\mu}}{c_i}}
   + \sqrt{\bigg(\delta\lambda_i^* + (1-\delta)\sqrt{\frac{1+\frac{\tilde{\lambda}_{-i}(\delta)}{\mu}}{c_i}} \bigg)^2 - \frac{1+\frac{\tilde{\lambda}_{-i}(\delta)}{\mu}}{c_i}}  = 0, 
\end{align}
where $\tilde{\lambda}_{-i}(\delta) = \sum\limits_{j =1}^N\tilde{\lambda}_{j}(\delta)-\tilde{\lambda}_{i}(\delta)$. Here we have $\lambda_i^{**}<\tilde{\lambda}_i(\delta) \leq \lambda_i^*$ for all $i \in \{1, \cdots, N\}$. 
As $\delta \to 0$, the proposed $\big(\tilde{\lambda}_1(\delta), \tilde{\lambda}_2(\delta), \cdots, \tilde{\lambda}_N(\delta)\big)$ approach $(\lambda_1^*, \lambda_2^*, \cdots, \lambda_N^*)$ in (3),  and the repeated game degenerates to one-shot game. As $\delta$ increases, cooperation profile $\big(\tilde{\lambda}_1(\delta), \tilde{\lambda}_2(\delta), \cdots, \tilde{\lambda}_N(\delta)\big)$ decrease and the competition mitigates.
\end{proposition}
The proof is given in Appendix E.
\begin{figure}
\centering
 \includegraphics[height=2.5in, width=3in]{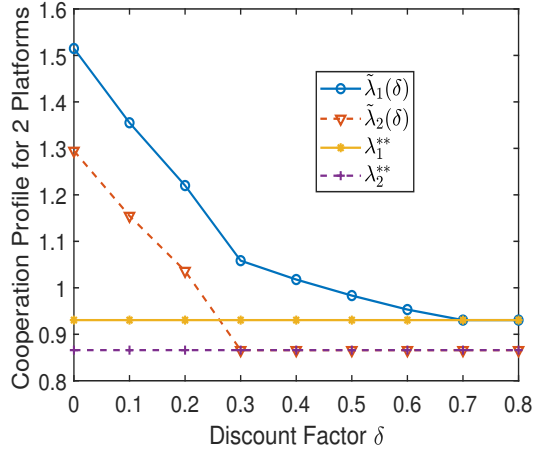}
\caption{Cooperation profile $(\tilde{\lambda}_1(\delta), \tilde{\lambda}_2(\delta))$ of the trigger mechanism versus the discount factor $\delta$,  as compared to social optimizers $(\lambda_1^{**},\lambda_2^{**})$. Here $N = 2$, $\delta^{th}_2=0.3$ and $\delta^{th}_1=0.7$ under parameters $c_1 = 1$, $c_2 = 1.5$ and $\mu = 1$. }

 \end{figure}
Figure 2 shows an illustrative example of $N =2$ platforms. It shows how cooperation profile $(\tilde{\lambda}_1(\delta), \tilde{\lambda}_2(\delta))$ under our trigger mechanism of non-monetary punishment changes with discount factor $\delta$ in all the three $\delta$ regimes. In small $\delta$ regime ($\delta<0.3$),  both $(\tilde{\lambda}_1(\delta), \tilde{\lambda}_2(\delta))$ decrease with $\delta$ until $\delta^{th}_2=0.3$ with $\tilde{\lambda}_2(\delta)=\lambda_2^{**}$. In medium $\delta$ regime ($0.3\leq\delta<0.7$), only $\tilde{\lambda}_1(\delta)$ decreases with $\delta$ until $\delta^{th}_1=0.7$. Finally, in large $\delta$ regime ($\delta\geq0.7$), the profile always equals $(\lambda_1^{**}, \lambda_2^{**})$. The results are consistent with Propositions \Romannum{3}.2-\Romannum{3}.4.  

Under our optimized  trigger mechanism of non-monetary punishment, one may wonder how the efficiency loss due to platform competition changes with discount factor $\delta$ in all the three $\delta$ regimes.
Given the symmetric cost setting  {($c_1 = c_2 = \cdots = c_N$)}, $\delta_1^{th}=\cdots=\delta^{th}_N$ and there are only small and large $\delta$ regimes. In this case, we manage to analytically derive the following result.

\begin{corollary}
 Given {$c_1 = c_2 = \cdots = c_N$} under complete information, the ratio between the social costs under the trigger mechanism of non-monetary punishment and the social optimum decreases with $\delta$ until $\delta=\delta_1^{th}=\cdots=\delta^{th}_N$ and keeps constant 1 since then. 
\end{corollary}
\begin{figure}
\centering
 \includegraphics[height=2.5in, width=3in]{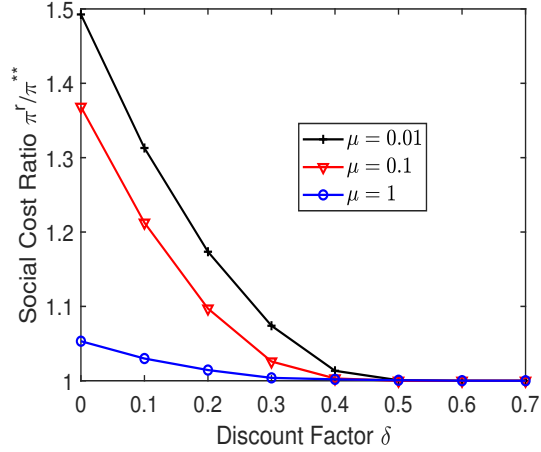}
\caption{Ratio between the social costs under our trigger-and-punishment mechanism and the social optimum under complete information. We fix $N =2$, $c_1=1$, $c_2 = 1.5$ and change bandwidth $\mu$ and $\delta$.}

 \end{figure}
\begin{proof}
As platforms perform the same in each round of the repeated game under our mechanism and the social optimum, it is enough to examine the social cost ratio in one shot.  Recall that under the social optimum, {$\lambda_1^{**} = \lambda_2^{**} = \cdots = \lambda_N^{**} = \frac{1}{\sqrt{c_1}}$} according to (3) and {$c_1=c_2=\cdots=c_N$}, the corresponding social cost is
 \begin{align}
     \pi^{**} = \sum\limits_{i=1}^N\bigg(\frac{1}{\lambda_i^{**}}\bigg(1+\sum\limits_{j = 1}^N\lambda_j^{**}/\mu\bigg) + c_i\lambda_i^{**}\bigg) \nonumber = 2N\sqrt{c_1} + 2N/\mu,\nonumber
 \end{align}
which is independent of $\delta$. Under our optimal trigger mechanism, {$\tilde{\lambda}_1(\delta) = \tilde{\lambda}_2(\delta) = \cdots = \tilde{\lambda}_N(\delta)$ decreases with $\delta$ until $\delta=\delta_1^{th}=\cdots=\delta^{th}_N$ and keeps constant since then}, and the corresponding social cost is
 \begin{align}
     \pi^{r} = \sum\limits_{i=1}^N\bigg(\frac{1}{\tilde{\lambda}_i(\delta)}\bigg(1+\sum\limits_{j = 1}^N\tilde{\lambda}_{j}(\delta)/\mu\bigg) + c_i\tilde{\lambda}_i(\delta) \bigg) \nonumber = N c_1\tilde{\lambda}_1(\delta) + N/\tilde{\lambda}_1(\delta) + 2N/\mu. 
 \end{align}
In the large $\delta$ regime, $\pi^r=\pi^{**}$ and we only need to examine the regime of small $\delta < \delta^{th}_1$, where $\tilde{\lambda}_1(\delta) > \lambda_1^{**}  = \frac{1}{\sqrt{c_1}}$ always holds according to (19). By taking the first derivative of $\pi^{r}$ over $\tilde{\lambda}_1(\delta)$, we have
{
 \begin{align}
 \frac{d\pi^{r}}{d\tilde{\lambda}_1(\delta)} = N c_1 - \frac{N}{\tilde{\lambda}_1^2(\delta)} > 0, \nonumber
 \end{align}
 }
 due to $\tilde{\lambda}_1(\delta) > \frac{1}{\sqrt{c_1}}$. Therefore, 
$\pi^{r}$ increases with $\tilde{\lambda}_1(\delta)$. According to Proposition \Romannum{3}.4,  $\tilde{\lambda}_1(\delta)$  decreases with $\delta$ for $\delta < \delta^{th}_1$, thus $\pi^{r} $ decreases with $\delta$ for $\delta < \delta^{th}_1$. Since $\pi^{**}$ is a constant with $\delta$, we have the ratio $\pi^{r}/\pi^{**}$ decreases with $\delta < \delta^{th}_1$ and keeps 1 after $\delta \geq \delta^{th}_1$.
\end{proof}

 Figure 3 further examines the asymmetric cost $c_1<c_2$ for the two-platform case. We can see that the the ratio between
social costs under the trigger mechanism and the social optimum still decreases with $\delta$, which is consistent with Corollary \Romannum{3}.4.1. As bandwidth $\mu$ increases, we expect smaller social cost ratio or smaller efficiency loss, as the two platforms' competition over bandwidth mitigates given more resource.

\section{Approximate trigger-and-punishment under Incomplete Information}

To remedy the huge inefficiency with PoA$=\infty$ in Proposition \Romannum{2}.4, we want to stimulate cooperation {among the $N$} platforms to approach the social optimum under incomplete information. As introduced in Section \Romannum{2}.B, platform 1's sampling cost realization in each instance is unknown to platform $i\in\{2, \cdots, N\}$. Similar to the complete information in Section \Romannum{3}, we propose to use an infinitely repeated game, where all the platforms will simultaneously play the Bayesian game for infinitely many  rounds  with  discount factor $\delta \in (0, 1)$.
Without any trigger mechanism of non-monetary punishment, each platform will still behave the same as $\left(\left(\lambda_{1}^{*}\left(c_{H}\right), \lambda_{1}^{*}\left(c_{L}\right)\right), \lambda_{2}^{*}, \cdots, \lambda_{N}^{*}\right)$ in (11)-(13) in each round. However, we cannot employ our non-forgiving trigger mechanism under complete information in Definition \Romannum{3}.1, by using the social optimal cooperation profile $\left(\left(\tilde{\lambda}_{1}\left(c_{L}, \delta\right), \tilde{\lambda}_{1}\left(c_{H}, \delta\right)\right), \tilde{\lambda}_{2}(\delta), \cdots, \tilde{\lambda}_{N}(\delta)\right)$. The reason is that under incomplete information, the other platforms cannot tell in each round whether platform $1$'s cost $c_1$ is $c_H$ or $c_L$ and platform $1$ can choose $\tilde{\lambda}_1(c_L)$ when $c_1=c_H$ without triggering any punishment. Even if $\delta$ is large enough to allow $(\tilde{\lambda}_1(c_L,\delta), \tilde{\lambda}_1(c_H,\delta))$=$(\lambda_1^{**}(c_L), \lambda_1^{**}(c_H))$, the following lemma shows that platform $1$ may not comply. 

 \begin{lemma}
Given the perfect cooperation profile $(\lambda_1^{**}(c_L),\lambda_1^{**}(c_H))$ for platform 1 under sufficiently large $\delta$,  platform 1 can still deviate from $\lambda_1^{**}(c_H)$ to $\lambda_1^{**}(c_L)$ when $c_1 = c_H$. 
\end{lemma}
The proof is given in Appendix F.
Once choosing between $\lambda_1^{**}(c_L)$ and $\lambda_1^{**}(c_H)$ in each round of the repeated game, platform 1 will not trigger any punishment. When  $c_1=c_H$, platform {$i\in\{2, \cdots, N\}$ under-samples with $\lambda_i^{**}$} by considering platform 1's average cost under incomplete information, and platform $1$ can take the information advantage to sample at high rate $\lambda_1^{**}(c_L)$ by using low AoI to justify its high sample cost. 

To negate information advantage of platform 1 under incomplete information, we next propose to blindly use platform 1's (deterministic) average cost to design its cooperation profile. That is, we recommend an approximate term $\tilde{\lambda}_1(\delta)$ to platform 1 all the time, without alternating between precise terms $\tilde{\lambda}_1(c_L, \delta)$ and $\tilde{\lambda}_1(c_H, \delta)$ over time to give platform 1 freedom to cheat. 

\begin{definition}
\emph{Our approximate trigger mechanism of indirect punishment under incomplete information is as follows:}
 \begin{itemize}
    \item \emph{In each round, all the platforms follow $\left(\tilde{\lambda}_{1}(\delta), \tilde{\lambda}_{2}(\delta), \cdots, \tilde{\lambda}_{N}(\delta)\right)$ to sample as approximate cooperation profile if none was ever detected to deviate from this profile in the past.}
    \item \emph{Once a deviation was found in the past, all the platforms will keep playing the equilibrium profile $\left(\left(\lambda_{1}^{*}\left(c_{H}\right), \lambda_{1}^{*}\left(c_{L}\right)\right), \lambda_{2}^{*}, \cdots, \lambda_{N}^{*}\right)$ in (11)-(13) forever as punishment.} 
\end{itemize}

\end{definition}
One can imagine that even if $\delta$ is sufficiently large, this approximate cooperation profile is still different from social optimizers $\left(\left(\lambda_{1}^{* *}\left(c_{H}\right), \lambda_{1}^{* *}\left(c_{L}\right)\right), \lambda_{2}^{* *}, \cdots, \lambda_{N}^{* *}\right)$ in (8)-(10) and there is inevitably some efficiency loss to avoid platform 1's cheating by using information advantage.  We next design the best cooperation profile $\left(\tilde{\lambda}_{1}(\delta), \tilde{\lambda}_{2}(\delta), \cdots, \tilde{\lambda}_{N}(\delta)\right)$ according to any value of non-trivial $\delta$ and minimize the involved inefficiency.

\subsection{Approximate  Cooperation Profile Design for Large $\delta$ Regime}

Under incomplete information,  platform 1 will behave indifferently no matter $c_1=c_H$ or $c_1=c_L$ in the repeated game. Then we can revise the social optimum in (8)-(10) by treating platform 1's cost constant as $\hat{c}_1:=p_Hc_H+(1-p_H)c_L$ deterministically. Then we approximate the social optimum as unique solutions to: 
\begin{align}
    \hat{\lambda}_{1}&=\sqrt{\frac{1+\sum\limits_{i=2}^{N} \hat{\lambda}_{i} / \mu}{p_{H} c_{H}+\left(1-p_{H}\right) c_{L}+\frac{1}{\mu} \sum\limits_{i=2}^{N} \frac{1}{\hat{\lambda}_{i}}}}, \\
   \hat{\lambda}_{i}&=\sqrt{\frac{1+\bigg(\sum\limits_{j = 1}^N \hat{\lambda}_{j}-\hat{\lambda}_i\bigg) / \mu}{c_{i}+\frac{1}{\mu}\bigg( \sum\limits_{j = 1}^N \frac{1}{\hat{\lambda}_{j}}-\frac{1}{\hat{\lambda}_i}\bigg)}}, \;\; i \in\{ 2, \cdots, N\}.
\end{align}
By comparing (20)-(21) with (8)-(10), we have the following result.
\begin{lemma}
Using approximation to smooth out sampling variation of platform 1, all the platforms will under-sample as compared to the social optimum. That is, 
\begin{align}
   \hat{\lambda}_1 \leq p_H\lambda_1^{**}(c_H) + (1-p_H)\lambda_1^{**}(c_L), \;\; \hat{\lambda}_i &\leq \lambda_i^{**}, \;\; i\in\{2, \cdots, N\}. \nonumber 
\end{align}
\end{lemma}

Given $(\hat{\lambda}_1, \hat{\lambda}_2, {\cdots, \hat{\lambda}_N})$ in (20)-(21) are attainable now and any platform's deviation from them will clearly bring itself in a larger long-term cost. We can analyze the no-deviation condition to reverse-engineer the feasible regime of large $\delta$ for enabling $(\hat{\lambda}_1, \hat{\lambda}_2, {\cdots, \hat{\lambda}_N})$ in the first place. 

When $c_1 = c_L$, if platform 1 chooses to deviate in the first round to save the immediate cost $\pi_{1}\left(\lambda_{1}, \sum\limits_{i=2}^{N} \hat{\lambda}_{i}\right)$ in (1) without any time discount, its optimal deviation or best response to {$\sum\limits_{i=2}^{N}\hat{\lambda}_i$} is $\lambda_{1}=\sqrt{\frac{1+\sum\limits_{i=2}^{N} \hat{\lambda}_{i} / \mu}{c_{L}}}$ according to (3). Its (discounted) long-term cost objective over all time stages is 
 \begin{align}
   \hat{\Pi}_{1}\left(c_{L}\right)&=\pi_{1}\left(\sqrt{\frac{1+\sum\limits_{i=2}^{N} \hat{\lambda}_{i} / \mu}{c_{L}}}, \sum\limits_{i=2}^{N} \hat{\lambda}_{i}\right)+\frac{\delta}{1-\delta} \pi_{1}\left(\left(\lambda_{1}^{*}\left(c_{L}\right), \lambda_{1}^{*}\left(c_{H}\right)\right), \lambda_{2}^{*}, \cdots, \lambda_{N}^{*}\right). \nonumber
\end{align}
Otherwise, it will obtain the following cost without any deviation,
    \begin{align}
   \Pi_{1}\left(c_{L}\right)&=\pi_{1}\left(\hat{\lambda}_{1}, \hat{\lambda}_{2}, \cdots, \hat{\lambda}_{N} | c_{1}=c_{L}\right)+\frac{\delta}{1-\delta} \pi_{1}\left(\hat{\lambda}_{1}, \hat{\lambda}_{2}, \cdots, \hat{\lambda}_{N}\right). \nonumber
\end{align} 
To ensure that platform 1 never deviates when $c_1=c_L$ in the repeated game, we require $\hat{\Pi}_1(c_L)\geq \Pi_1(c_L)$, or simply $\delta\geq \hat{\delta}_1^{th}(c_L)$ with 
\begin{table}[H]
\footnotesize{
  \begin{align}
&\hat{\delta}^{th}_1(c_L) :=\frac{c_{L} \hat{\lambda}_{1}-2 \sqrt{\left(1+\sum\limits_{i=2}^{N} \hat{\lambda}_{i} / \mu\right) c_{L}}+\left(1+\sum\limits_{i=2}^{N} \hat{\lambda}_{i} / \mu\right) / \hat{\lambda}_{1}}{\left(c_{L}-\hat{c}_{1}\right) \hat{\lambda}_{1}+2 \sqrt{1+\sum\limits_{i=2}^{N} \lambda_{i}^{*} / \mu}\left(p_{H} \sqrt{c_{H}}+\left(1-p_{H}\right) \sqrt{c_{L}}\right)-2 \sqrt{\left(1+\sum\limits_{i=2}^{N} \hat{\lambda}_{i} / \mu\right) c_{L}}}.
\end{align}}
\vspace{-30pt}
\end{table}
Similarly, we require the following to ensure no deviation when $c_1 = c_H$:
\begin{table}[H]
\footnotesize{
  \begin{align}
&\delta\geq \hat{\delta}^{th}_1(c_H) :=\frac{c_{H} \hat{\lambda}_{1}-2 \sqrt{\left(1+\sum\limits_{i=2}^{N} \hat{\lambda}_{i} / \mu\right) c_{H}}+\left(1+\sum\limits_{i=2}^{N} \hat{\lambda}_{i} / \mu\right) / \hat{\lambda}_{1}}{\left(c_{H}-\hat{c}_{1}\right) \hat{\lambda}_{1}+2 \sqrt{1+\sum\limits_{i=2}^{N} \lambda_{i}^{*} / \mu}\left(p_{H} \sqrt{c_{H}}+\left(1-p_{H}\right) \sqrt{c_{L}}\right)-2 \sqrt{\left(1+\sum\limits_{i=2}^{N} \hat{\lambda}_{i} / \mu\right) c_{H}}}.
\end{align}}
\vspace{-30pt}
\end{table}
Given platform 1 always chooses $\hat{\lambda}_1$, we also require the following for platform $i\in\{2,\cdots,N\}$ to follow $\hat{\lambda}_i$:   

     \begin{align}
      \delta \geq \hat{\delta}_{i}^{t h} :=\frac{\sqrt{\frac{c_{i}+\frac{1}{\mu} \bigg(\sum\limits_{j = 1}^N \frac{1}{\hat{\lambda}_{j}}-\frac{1}{\hat{\lambda}_i}\bigg)}{c_{i}}}+\sqrt{\frac{c_{i}}{c_{i}+\frac{1}{\mu}\bigg( \sum\limits_{j = 1}^N \frac{1}{\hat{\lambda}_{j}}-\frac{1}{\hat{\lambda}_i}\bigg)}}-2}{2 \sqrt{\frac{1+\left(p_{H} \lambda_{1}^{*}\left(c_{H}\right)+\left(1-p_{H}\right) \lambda_{1}^{*}\left(c_{L}\right)+\sum\limits_{j = 2}^N \lambda_{j}^{*}-\lambda_i^*\right) / \mu}{1+\bigg(\sum\limits_{j=1}^N \frac{1}{\hat{\lambda}_{j}}-\frac{1}{\hat{\lambda}_i}\bigg) / \mu}}-2}.
\end{align}
 Recall that {$p_Hc_H + (1-p_H)c_L \leq c_2 \leq \cdots \leq c_N$}, we have {$\hat{\delta}^{th}_N \leq \cdots \leq \hat{\delta}^{th}_{2} \leq \max\{ \hat{\delta}^{th}_1(c_H)$, $  \hat{\delta}^{th}_1(c_L) \}$ $:=  \hat{\delta}^{th}_1$}. Yet note that $\min\{ \hat{\delta}^{th}_1(c_H)$, $  \hat{\delta}^{th}_1(c_L) \}$ may or may not be larger than $\hat{\delta}^{th}_{2}$.
 
\begin{proposition}[Large $\delta$ Regime]
Under incomplete information, if $\delta \geq \hat{\delta}^{th}_1 = \max\{ \hat{\delta}^{th}_1(c_H)$, $  \hat{\delta}^{th}_1(c_L) \}$, {all the $N$} platforms will follow the approximate cooperation profile $(\tilde{\lambda}_1(\delta),\tilde{\lambda}_2(\delta), \cdots$, $ \tilde{\lambda}_N(\delta))=(\hat{\lambda}_1,\hat{\lambda}_2, \cdots$, $\hat{\lambda}_N)$ in (20)-(21) all the time, without triggering punishment $((\lambda_{1}^{*}\left(c_{H}\right)$, $\lambda_{1}^{*}\left(c_{L}\right)), \lambda_{2}^{*}, \cdots, \lambda_{N}^{*})$ in (11)-(13).
\end{proposition}
When all the platforms have the same average costs, i.e., $p_Hc_H+ (1-p_H)c_L=c_{2} =\cdots= c_N$, we can analytically prove the following proposition. 
\begin{proposition}
 Given symmetric costs $p_Hc_H+ (1-p_H)c_L=c_{2} =\cdots= c_N$ among the platforms, the approximation ratio achieved by our trigger mechanism with profile  $(\hat{\lambda}_1,\hat{\lambda}_2, \cdots, \hat{\lambda}_N)$ in (20)-(21) is $\frac{N}{N-1}$ as compared to the social optimum with $\left(\left(\lambda_{1}^{**}\left(c_{H}\right), \lambda_{1}^{**}\left(c_{L}\right)\right), \lambda_{2}^{**}, \cdots, \lambda_{N}^{**}\right)$ in (8)-(10). The mechanism's performance improves as we have more incumbent platforms with known cost information. 
\end{proposition}
The proof is given in Appendix G. Given only platform 1 with hidden cost, relatively we face less information uncertainty as total platform number $N$ increases. 

\subsection{{Approximate Cooperation Profile Design for Medium $\delta$ Regime}}
 
If $\hat{\delta}^{th}_{j} \leq \delta < \hat{\delta}^{th}_{j-1}$, where $j \in\{ 2, \cdots, N\}$, only platform $k \in \{j, \cdots, N\}$ will still follow perfect approximate profile $\hat{\lambda}_k$ in (20). Yet platform 1 will deviate when $c_1=c_L$ given $\delta<\max\{ \hat{\delta}^{th}_1(c_H),  \hat{\delta}^{th}_1(c_L) \}$, and platform $i \in \{2, \cdots, j-1\}$ will also deviate from $\hat{\lambda}_i$, requiring us to design new $\tilde{\lambda}_1(\delta)$ and $\tilde{\lambda}_i(\delta)$. Similar to (22), (23) and (24), we need to ensure the platform 1's long-term cost does not change after the best immediate deviation no matter whether $c_1=c_L$ or $c_1=c_H$, and ensure platform $i$'s long-term cost does not change after the best immediate deviation. The we have the following.

\begin{proposition}[{Medium $\delta$ Regime}]
In the repeated game under incomplete information, if $\hat{\delta}^{th}_j \leq \delta < \hat{\delta}^{th}_{j-1}$ for some $j \in\{ 2, \cdots, N\}$, all the $N$ platforms will always follow the cooperation profile below without deviating to trigger $\left(\left(\lambda_{1}^{*}\left(c_{H}\right), \lambda_{1}^{*}\left(c_{L}\right)\right), \lambda_{2}^{*}, \cdots, \lambda_{N}^{*}\right)$ in (11)-(13) as punishment: 
\begin{itemize}
    \item For platform $k \in \{j, \cdots, N\}$ with greater costs: $\tilde{\lambda}_k(\delta) = \hat{\lambda}_k$ in (20)-(21).
    \item For platform 1 and platform $i\in \{2, \cdots, j-1\}$, their cooperation profile $\big(\tilde{\lambda}_1(\delta), \cdots, \tilde{\lambda}_{j-1}(\delta)\big)$ to follow are the unique solutions to: 
    \vspace{-10pt}
    \begin{table}[H]
        \small{
            \begin{align}     
 \tilde{\lambda}_{1}(\delta)=&\max \Bigg\{\frac{M_{L}-\sqrt{M_{L}^{2}-\left(\delta \hat{c}_{1}+(1-\delta) c_{L}\right)\left(1+\left(\sum\limits_{k = 2}^{j-1} \tilde{\lambda}_{k}(\delta)+\sum\limits_{k=j}^N \hat{\lambda}_{k}\right) / \mu\right)}}{\delta \hat{c}_{1}+(1-\delta) c_{L}}, \nonumber \\
 &\frac{M_{H}-\sqrt{M_{H}^{2}-\left(\delta \hat{c}_{1}+(1-\delta) c_{H}\right)\left(1+\left(\sum\limits_{k = 2}^{j-1} \tilde{\lambda}_{k}(\delta)+\sum\limits_{k=j}^N \hat{\lambda}_{k}\right) / \mu\right)}}{\delta \hat{c}_{1}+(1-\delta) c_{H}} \Bigg\}, \nonumber 
 \end{align}
 \begin{align} 
 \tilde{\lambda}_{i}(\delta)-\frac{M_{i}-\sqrt{M_{i}^{2}-c_{i}\left(1+\left(\sum\limits_{k = 1}^{j-1} \tilde{\lambda}_{k}(\delta)-\tilde{\lambda}_{i}(\delta)+\sum\limits_{k=j}^N \hat{\lambda}_{k}\right) / \mu\right)}}{c_{i}}=0,\nonumber
        \end{align}
        }
        \vspace{-30pt}
    \end{table}
where
\small{\begin{align}
M_{L}&=\delta \sqrt{1+\sum\limits_{i=2}^{N} \lambda_{i}^{*} / \mu}\left(p_{H} \sqrt{c_{H}}+\left(1-p_{H}\right) \sqrt{c_{L}}\right)+(1-\delta) \sqrt{\left(1+\left(\sum\limits_{k = 2}^{j-1} \tilde{\lambda}_{k}(\delta)+\sum\limits_{k=j}^N \hat{\lambda}_{k}\right) / \mu\right) c_{L}}, \nonumber  \\
M_{H}&=\delta \sqrt{1+\sum\limits_{i=2}^{N} \lambda_{i}^{*} / \mu}\left(p_{H} \sqrt{c_{H}}+\left(1-p_{H}\right) \sqrt{c_{L}}\right)+(1-\delta) \sqrt{\left(1+\left(\sum\limits_{k = 2}^{j-1} \tilde{\lambda}_{k}(\delta)+\sum\limits_{k=j}^N \hat{\lambda}_{k}\right) / \mu\right) c_{H}}, \nonumber \\
M_{i}&=\sqrt{c_{i}}(\delta \sqrt{1+(p_{H} \lambda_{1}^{*}\left(c_{H}\right)+\left(1-p_{H}\right) \lambda_{1}^{*}\left(c_{L}\right)+\sum\limits_{k = 2}^{N} \lambda_{k}^{*}-\lambda_i^* ) / \mu} \nonumber \\
&+(1-\delta) \sqrt{1+\left(\sum\limits_{k = 1}^{j-1} \tilde{\lambda}_{k}(\delta)-\tilde{\lambda}_{i}(\delta)+\sum\limits_{k=j}^N \hat{\lambda}_{k}\right) / \mu}). \nonumber
\end{align}}
\end{itemize}
\end{proposition}
The proof is given in Appendix H.

\subsection{Approximate Cooperation Profile Design for Small $\delta$ Regime}

If $\delta$ is smaller than the smallest threshold ${\hat{\delta}^{th}_N}$ among the platforms, no platform will follow cooperation profile $(\hat{\lambda}_1, \hat{\lambda}_2, {\cdots, \hat{\lambda}_N})$ in (20)-(21), and we need to redesign new $\big(\tilde{\lambda}_1(\delta), \tilde{\lambda}_2(\delta), \cdots$, $ \tilde{\lambda}_N(\delta)\big)$ jointly as functions of $\delta$. Similarly, we need to design the cooperation profile such that the platforms' long-term discounted costs do not change after the best immediate deviation. 
\begin{proposition}[Small $\delta$ Regime]
 In the repeated game under incomplete information, if  $\delta < \hat{\delta}^{th}_N $ with $i = N$ in (24), the $N$ platforms will always follow the cooperation profile $\big(\tilde{\lambda}_1(\delta), \tilde{\lambda}_2(\delta), \cdots, \tilde{\lambda}_N(\delta)\big)$ as unique solutions to
 \vspace{-10pt}
 \begin{table}[H]
    \small{
     \begin{align}     
 \tilde{\lambda}_{1}(\delta)=&\max \Bigg\{\frac{M_{L}^{\prime}-\sqrt{M_{L}^{\prime 2}-\left(\delta \hat{c}_{1}+(1-\delta) c_{L}\right)\left(1+\sum\limits_{i=2}^{N} \tilde{\lambda}_{i}(\delta) / \mu\right)}}{\delta \hat{c}_{1}+(1-\delta) c_{L}}, \nonumber \\
 &\frac{M_{H}^{\prime}-\sqrt{M_{H}^{\prime 2}-\left(\delta \hat{c}_{1}+(1-\delta) c_{H}\right)\left(1+\sum\limits_{i=2}^{N} \tilde{\lambda}_{i}(\delta) / \mu\right)}}{\delta \hat{c}_{1}+(1-\delta) c_{H}} \Bigg\}, \nonumber
 \end{align}
 \begin{align} 
 &\tilde{\lambda}_{i}(\delta)-\frac{M_{i}^{\prime}-\sqrt{M_{i}^{\prime 2}-c_{i}\left(1+\big(\sum\limits_{j = 1}^N \tilde{\lambda}_{j}(\delta)-\tilde{\lambda}_{i}(\delta)\big) / \mu\right)}}{c_{i}}=0, \nonumber
        \end{align}
    }
    \vspace{-30pt}
 \end{table}
where 
\small{\begin{align}
\begin{aligned} &M_{L}^{\prime} =\delta \sqrt{1+\sum\limits_{i=2}^{N} \lambda_{i}^{*} / \mu}\left(p_{H} \sqrt{c_{H}}+\left(1-p_{H}\right) \sqrt{c_{L}}\right)+(1-\delta) \sqrt{\left(1+\sum\limits_{i=2}^{N} \tilde{\lambda}_{i}(\delta) / \mu\right) c_{L}}, \\ 
&M_{H}^{\prime} =\delta \sqrt{1+\sum\limits_{i=2}^{N} \lambda_{i}^{*} / \mu}\left(p_{H} \sqrt{c_{H}}+\left(1-p_{H}\right) \sqrt{c_{L}}\right)+(1-\delta) \sqrt{\left(1+\sum\limits_{i=2}^{N} \tilde{\lambda}_{i}(\delta) / \mu\right) c_{H}}, \\ 
&M_{i}^{\prime} =\sqrt{c_{i}}\left(\delta \sqrt{1+(p_{H} \lambda_{1}^{*}\left(c_{H}\right)+\left(1-p_{H}\right) \lambda_{1}^{*}\left(c_{L}\right)+\sum\limits_{j = 2}^{N} \lambda_{j}^{*}-\lambda_i^* ) / \mu}+(1-\delta) \sqrt{1+\big(\sum\limits_{j = 1}^{N} \tilde{\lambda}_{j}(\delta) - \tilde{\lambda}_{i}(\delta)\big) / \mu}\right). \end{aligned} \nonumber
\end{align}}
\end{proposition}
The proof is given in Appendix I.
Figure 4 shows an illustrative example of $N=2$ platforms, where the approximate cooperation profile $\big(\tilde{\lambda}_1(\delta), \tilde{\lambda}_2(\delta)\big)$ in Propositions \Romannum{4}.4, \Romannum{4}.6, \Romannum{4}.7 under our trigger mechanism of non-monetary punishment changes with discount factor $\delta$ in all the three $\delta$ regimes. Here the mean cost $\hat{c}_1$ of platform 1 is less than that of platform 2 with $\hat{\delta}^{th}_1(c_L) = \max\{\hat{\delta}^{th}_1(c_L), \hat{\delta}^{th}_1(c_H)\}$ and $\hat{\delta}^{th}_2<\hat{\delta}^{th}_1(c_L)$. In small $\delta$ regime, both $(\tilde{\lambda}_1(\delta),\tilde{\lambda}_2(\delta))$ decrease with $\delta$ until $\hat{\delta}^{th}_2=0.3$ with $\tilde{\lambda}_2(\delta)= \hat{\lambda}_2$ ideally. In medium $\delta$ regime, only $\tilde{\lambda}_1(\delta)$ decreases with $\delta$ till $\hat{\delta}^{th}_1(c_L)$ = 0.7. 
Finally, in large $\delta$ regime, the profile eventually equals $(\hat{\lambda}_1, \hat{\lambda}_2)$. As $\hat{c}_1=19$ is slightly smaller than $ c_2=20$ here, the final profile $\hat{\lambda}_1$ is close to $\hat{\lambda}_2$. 
  \begin{figure}
 \centering
 \includegraphics[height=2.5in, width=3in]{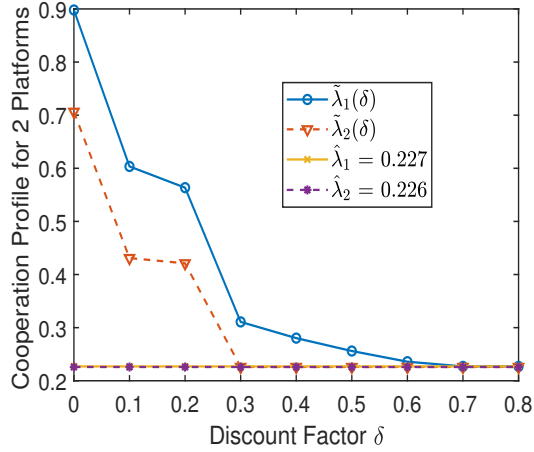}
\caption{Cooperation profile $(\tilde{\lambda}_1(\delta), \tilde{\lambda}_2(\delta))$ of the trigger mechanism versus discount factor $\delta$, as compared to approximate social optimizers $(\hat{\lambda}_1, \hat{\lambda}_2)$ in (20)-(21). Here we have $N =2$, $\hat{\delta}^{th}_2=0.3$ and $\hat{\delta}^{th}_1(c_L)=0.7$ under parameters $c_H = 100$, $c_L = 10$, $p_H = 0.1$, $\hat{c}_1=19$,  $c_2 = 20$ and $\mu = 0.1$ with $\hat{c}_1 <c_2$. }
\vspace{-30pt}
 \end{figure}

  \begin{figure*}[!t]
\centering
\subfloat[][Social cost ratio between equilibrium and optimum]{\includegraphics[width=3in]{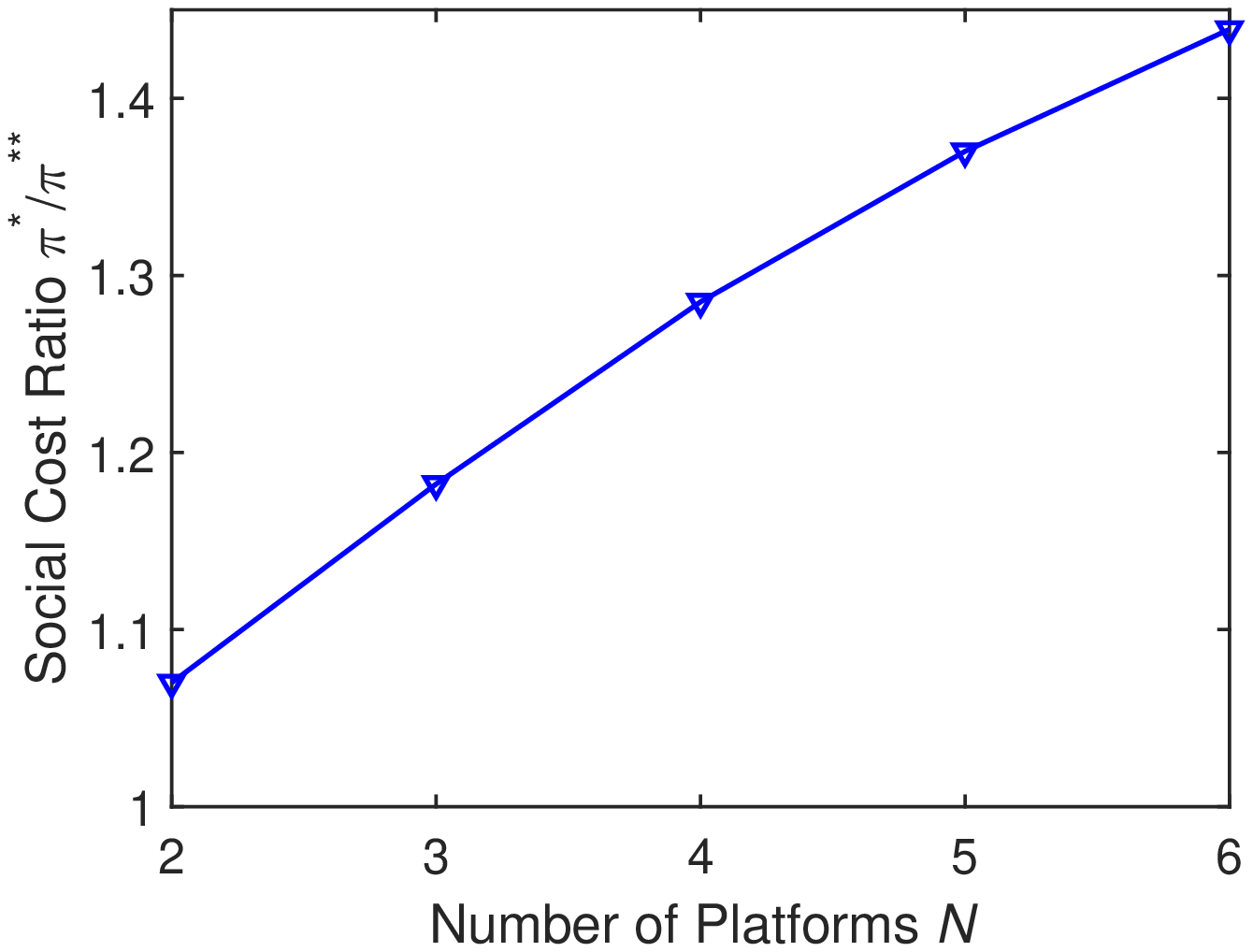}%
\label{}}
\hfil
\subfloat[][Social cost ratio between approximate mechanism and optimum]{\includegraphics[width=3in]{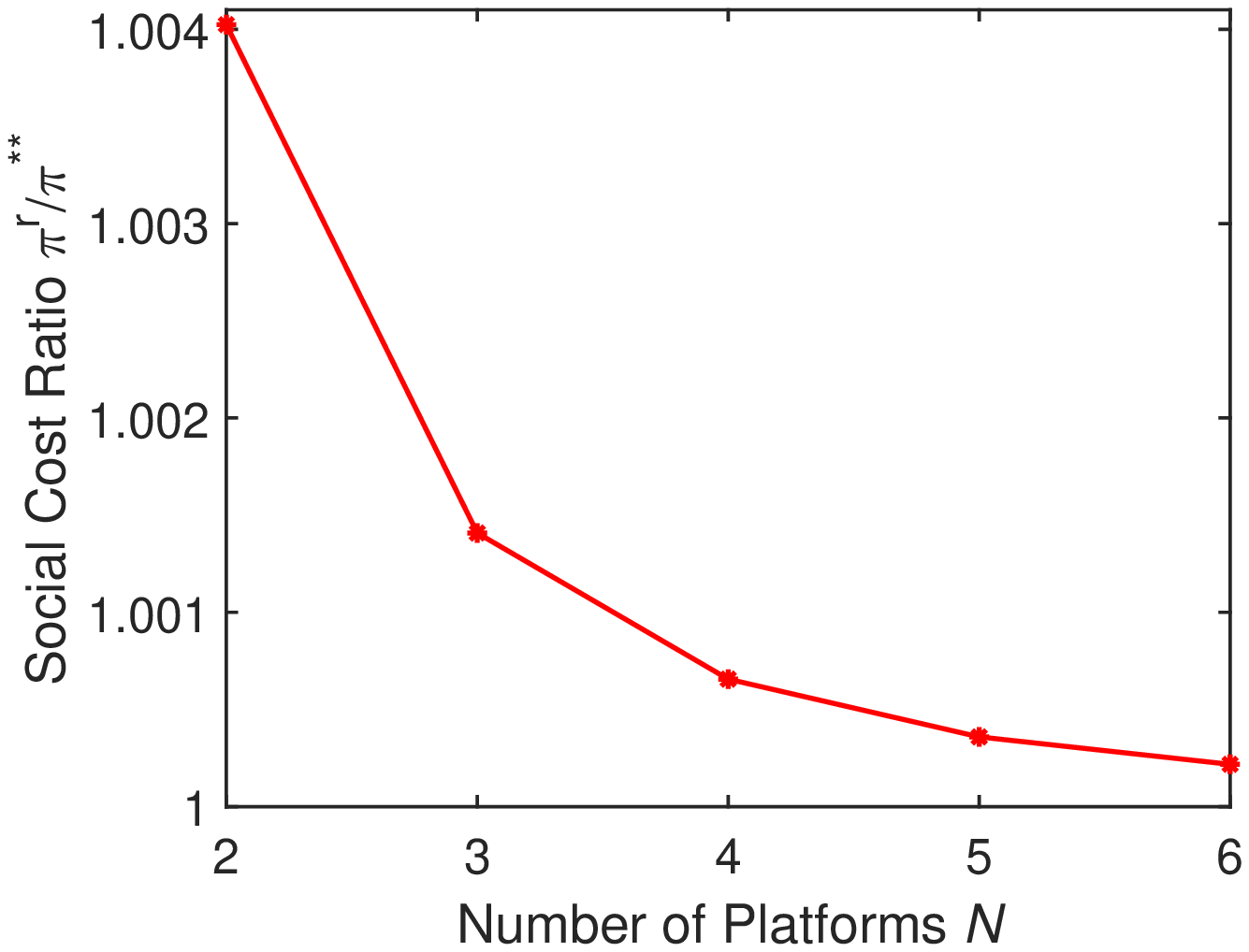}%
\label{}}
\caption{Empirical performance comparison between competition equilibrium, social optimum in Section \Romannum{2}.B, and our approximate mechanism here. Here, we set $c_H = 1.5$, $c_L = 0.5$, $p_H = 0.5$ for platform 1, and symmetric costs $c_2 = \cdots = c_6 = p_H c_H + (1-p_H)c_L$ for the other $N-1$ platforms, and $\mu = 1$.}
\label{fig_sim}
\vspace{-20pt}
\end{figure*}
Figure 5 considers an arbitrary number $N$ of platforms and empirically shows the social cost ratio between the approximate trigger mechanism (in large $\delta$ regime) and the social optimum in (8)-(10), $\pi^r/\pi^{**}$, by comparing to the social cost ratio between competition equilibrium in (11)-(13) and the social optimum, $\pi^*/\pi^{**}$ without any mechanism design. As $N$ increases, platforms compete more intensively to over-sample, thus the ratio $\pi^*/\pi^{**}$ increases with greater efficiency loss. However, our approximate mechanism only has mild efficiency loss. Given only platform 1 with hidden information, relatively we face less information uncertainty as the total platform number $N$ increases, and the approximate cooperation profile better approaches the social optimizers. Hence, ratio $\pi^r/\pi^{**}$ decreases. This empricial result is consistent with Proposition \Romannum{4}.5 in the worst case. Similar to Figure 3 in Section \Romannum{3}, with asymmetric unit sampling costs, our simulations show that social cost ratio between approximate mechanism and optimum under incomplete information also decreases with $\delta$ in small and medium $\delta$ regimes, and keeps constant in large $\delta$ regime.




\section{Conclusion}

In this paper we study the competition among online content platforms in AoI and bandwidth sharing, and they concern the freshness of their own updates on real-time information instead of the others'. When all the platforms know each other's sampling cost, 
we show that all the platforms over-sample and cause huge efficiency loss. To remedy the loss, we propose a trigger mechanism of non-monetary punishment in the repeated game to approach the social optimum. We also study the more challenging case where some newly joined platform can hide its cost information from the other incumbent platforms in the Bayesian game. Perhaps surprisingly, we show that this platform may get hurt by knowing more information. Accordingly, we  redesign the trigger-and-punishment mechanism to approach the social optimum by ensuring no cheating from the platform with more information. Extensive simulations show that the mechanism's performance improves as we have more incumbent platforms with known cost information.



%
\appendices

\section{Proof of Proposition \Romannum{2}.1}
\subsection{Proof of Social Optimizers and the Uniqueness}
To show (2) are solutions as social optimizers, note that $\pi(\lambda_1, \cdots, \lambda_N)$ is concave with each $\lambda_i$ due to $\frac{\partial^2 \pi(\lambda_1, \cdots, \lambda_N)}{\partial \lambda_i^2}\leq 0$, where $i\in\{1, \cdots, N\}$. By using the first-order condition, we have (2) as the solutions.

     Then we want to prove (2) has unique solutions with induction method. When $N = 2$,  (2) can be rewritten as 
     
      \begin{align}
    \lambda_1(\lambda_2) &= \sqrt{\frac{1+\frac{\lambda_2}{\mu}}{\frac{1}{\lambda_2\mu}+c_1}},  \\
    \lambda_2(\lambda_1) &= \sqrt{\frac{1+\frac{\lambda_1}{\mu}}{\frac{1}{\lambda_1\mu}+c_2}}. 
\end{align}
To show $\lambda_1(\lambda_2)$ in (25) is concave and strictly increasing in $\lambda_2$, we take the first and second derivatives of $\lambda_1(\lambda_2)$ as 
\begin{align}
    \lambda_1'(\lambda_2) &= \frac{1}{2\mu}\sqrt{\frac{\frac{1}{\lambda_2\mu}+c_1}{1+\frac{\lambda_2}{\mu}}}\cdot\frac{c_1+\frac{2}{\lambda_2\mu}+ \frac{1}{\lambda_2^2}}{(\frac{1}{\lambda_2\mu}+c_1)^2}, \nonumber \\
    \lambda_1''(\lambda_2) &= -\frac{\frac{5c_1+1}{\lambda_2^2} + \frac{2c_1+2}{\lambda_2\mu} + c_1 + \frac{1}{\lambda_2^4} + \frac{4c_1\mu}{\lambda_2^3}}{4\mu^2(\frac{1}{\lambda_2\mu}+c_1)^{\frac{5}{2}}(1+\frac{\lambda_2}{\mu})^{\frac{3}{2}}}. \nonumber
\end{align}
Since $\lambda_1''(\lambda_2) < 0$, $\lambda_1'(\lambda_2 = 0) \to \infty$  and $\lambda_1'(\lambda_2 \to \infty) = 0$, we know $\lambda_1'(\lambda_2) > 0$ and $\lambda_1(\lambda_2)$ is concave and strictly increasing in $\lambda_2$. Similarly we can show that $\lambda_2(\lambda_1)$ in (26) is concave and strictly increasing in $\lambda_1$. By substituting $\lambda_1(\lambda_2)$ in (25) into $\lambda_2(\lambda_1)$ in (26), we simplify (25) and (26) as the following equation:
\begin{align}
    \lambda_2- \sqrt{\frac{1+\frac{\lambda_1(\lambda_2)}{\mu}}{\frac{1}{\lambda_1(\lambda_2)\mu}+c_2}} = 0.
\end{align}
Since  $\sqrt{\frac{1+\frac{\lambda_1(\lambda_2)}{\mu}}{\frac{1}{\lambda_1(\lambda_2)\mu}+c_2}}$ in (27) is concave and strictly increasing in $\lambda_1(\lambda_2)$ and  $\lambda_1(\lambda_2)$ in (27) is concave and strictly increasing in $\lambda_2$, we obtain that $\sqrt{\frac{1+\frac{\lambda_1(\lambda_2)}{\mu}}{\frac{1}{\lambda_1(\lambda_2)\mu}+c_2}}$ is concave and strictly increasing in $\lambda_2$. To show (27) has only one positive solution, denote $g(\lambda_2) = \lambda_2 - \sqrt{\frac{1+\frac{\lambda_1(\lambda_2)}{\mu}}{\frac{1}{\lambda_1(\lambda_2)\mu}+c_2}}$, where we know $g(\lambda_2)$ is convex in $\lambda_2$ because of concavity of  $\sqrt{\frac{1+\frac{\lambda_1(\lambda_2)}{\mu}}{\frac{1}{\lambda_1(\lambda_2)\mu}+c_2}}$. Therefore we know $g'(\lambda_2)$ increases with $\lambda_2$. Since  
\begin{align}
    g'(\lambda_2 = 0) &= 1 -  \lambda_2'(0) < 0, \nonumber \\
    g'(\lambda_2 \to \infty) &=  1 - \lambda_2'(\infty)  > 0, \nonumber
\end{align}
there exists unique $a > 0$ satisfying $g'(a) = 0$, then $g(\lambda_2)$ decreases in $(0, a]$ and increases in $(a, \infty)$. Since $g(0) = 0$ and $g(\infty) \to \infty$, there exists unique $\lambda_2^{**} > 0$ satisfying $g(\lambda_2^{**}) = 0$, thus (27) has unique positive solution. We plot $g(\lambda_2)$ in Figure 6. Similarly (25) has unique positive solution $\lambda_1^{**}$. The social optimizers in (2) are unique.

Suppose that when $N = M-1$, (2) has unique solutions. With induction method, we need to prove when $N = M$, (2) has unique solutions. Similar to (25)-(26), we can rewrite $\lambda_i$ as a function of $\lambda_j$ and the $\lambda_i$ is concave and strictly increasing in each $\lambda_j$, where $i\in\{1, \cdots, M\}$ and $j \ne i$. If we introduce $\lambda_M$ as in (25)-(26) into other $\lambda_i$ as in (25)-(26), where $i\in\{1, \cdots, M-1\}$, we have $\lambda_i$ is still concave and strictly increasing in $\lambda_j$, where $j \in \{1, \cdots, M-1\}$ and $j \ne i$. Since we know when $N = M-1$, (2) has unique solutions. Then after introducing $\lambda_M$ as in (25)-(26) into other $\lambda_i$ as in (25)-(26), the new $M-1$ equations also have unique solutions. Then we prove that when $N = M$, (2) has unique solutions.
\begin{figure}
\centering
 \includegraphics[height=2.5in, width=3in]{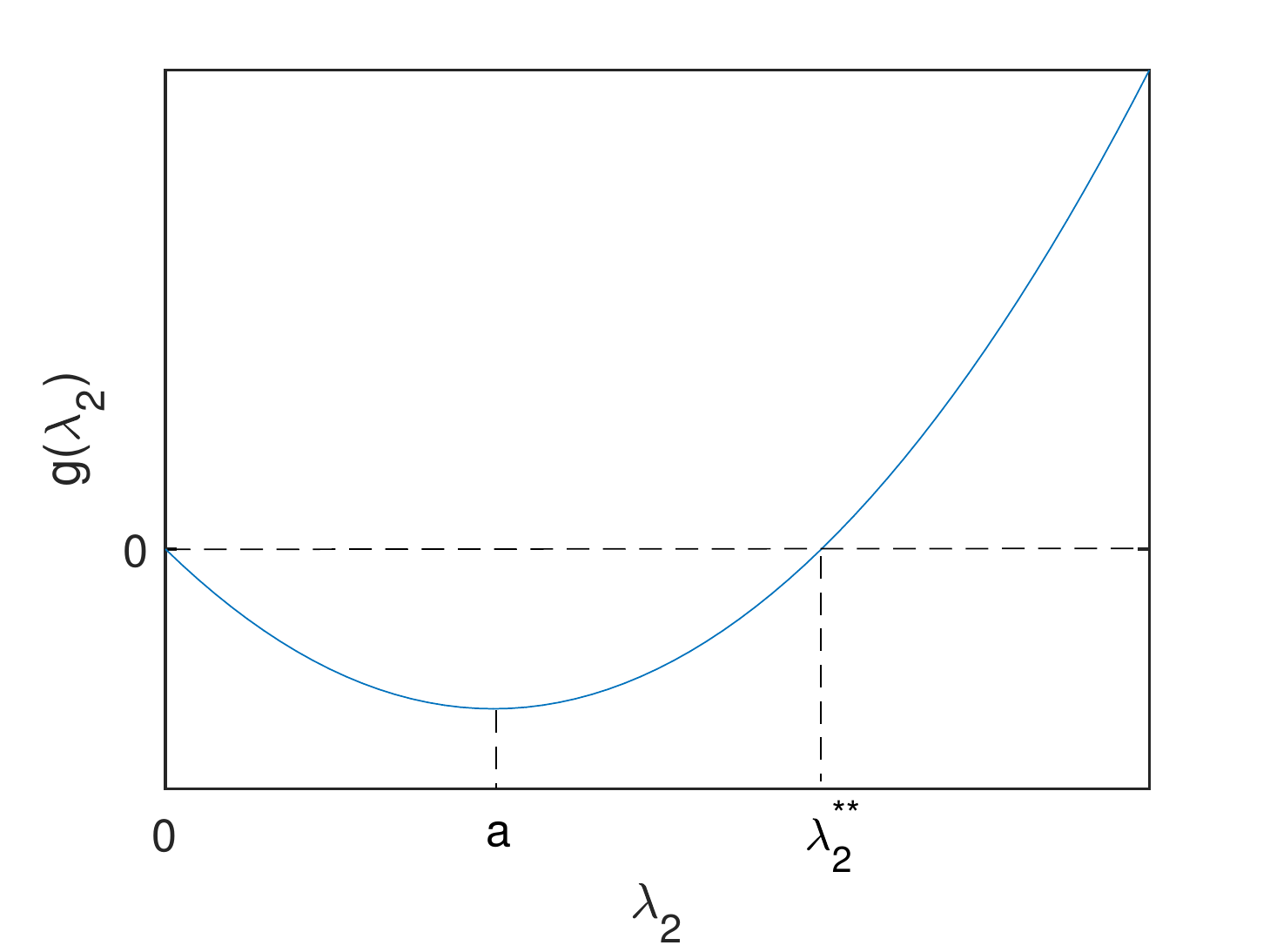}
\caption{$g(\lambda_2)$ versus $\lambda_2$ in Appendix A.A.}
 \end{figure}
\subsection{Proof of Competition Equilibrium and the Uniqueness}
To show (3) are solutions as equilibrium, note that each $\pi_i(\lambda_i, \lambda_{-i})$ is concave with each $\lambda_i$ due to $\frac{\partial^2 \pi_i(\lambda_i, \lambda_{-i})}{\partial \lambda_i^2}\leq 0$, where $i\in\{1, \cdots, N\}$. By using the first-order condition, we have (3) as the solutions.

 Then we want to prove (2) has unique solutions with induction method. When $N = 2$, (3) can be rewritten as
 
\begin{align}
    \lambda_1(\lambda_2) = \sqrt{\frac{1+\lambda_2/\mu}{c_1}},  \\
    \lambda_2(\lambda_1) = \sqrt{\frac{1+\lambda_1/\mu}{c_2}},
\end{align}
which are equivalent to the following equation:
\begin{align}
    \sqrt{\frac{1+\lambda_1/\mu}{c_2}}- \mu(c_1\lambda_1^2 - 1) = 0.
\end{align}
 Denote $f(\lambda_1) = \sqrt{\frac{1+\lambda_1/\mu}{c_2}}- \mu(c_1\lambda_1^2 - 1) $. To show $f(\lambda_1)$ only has one positive root, we check the first-order and second-order derivatives of $f(\lambda_1)$ as
\begin{align}
    f'(\lambda_1) &= \frac{1}{2 \sqrt{c_2} \mu}\sqrt{\frac{1}{1+\lambda_1/\mu}} - 2c_1 \mu\lambda_1, \nonumber \\
    f''(\lambda_1) &= -\frac{1}{4 \sqrt{c_2} \mu^2}\sqrt{\frac{1}{(1+\lambda_1/\mu)^{\frac{3}{2}}}} - 2c_1 \mu. \nonumber
\end{align}
Since $f''(\lambda_1) < 0$, $f'(\lambda_1)$ decreases with $\lambda_1$. Additionally, since $f'(0) > 0$ and $f'(\infty) < 0$, then $f'(\lambda_1)$ has exactly one positive root, denoted as $f'(b) = 0$. Then $f(\lambda_1)$ is increasing in $(0, b]$ and decreasing in $(b, \infty)$. Also since
\begin{align}
    f(0) &> 0, \nonumber \\
    f(\infty) &< 0, \nonumber
\end{align}
then $f(\lambda_1)$ has unique positive root $\lambda_1^*$ satisfying (30). We plot $f(\lambda_1)$ in Figure 7. Therefore $\lambda_2^*$ is unique according to (29) and (3) has unique solutions.
\begin{figure}
\centering
 \includegraphics[height=2.5in, width=3in]{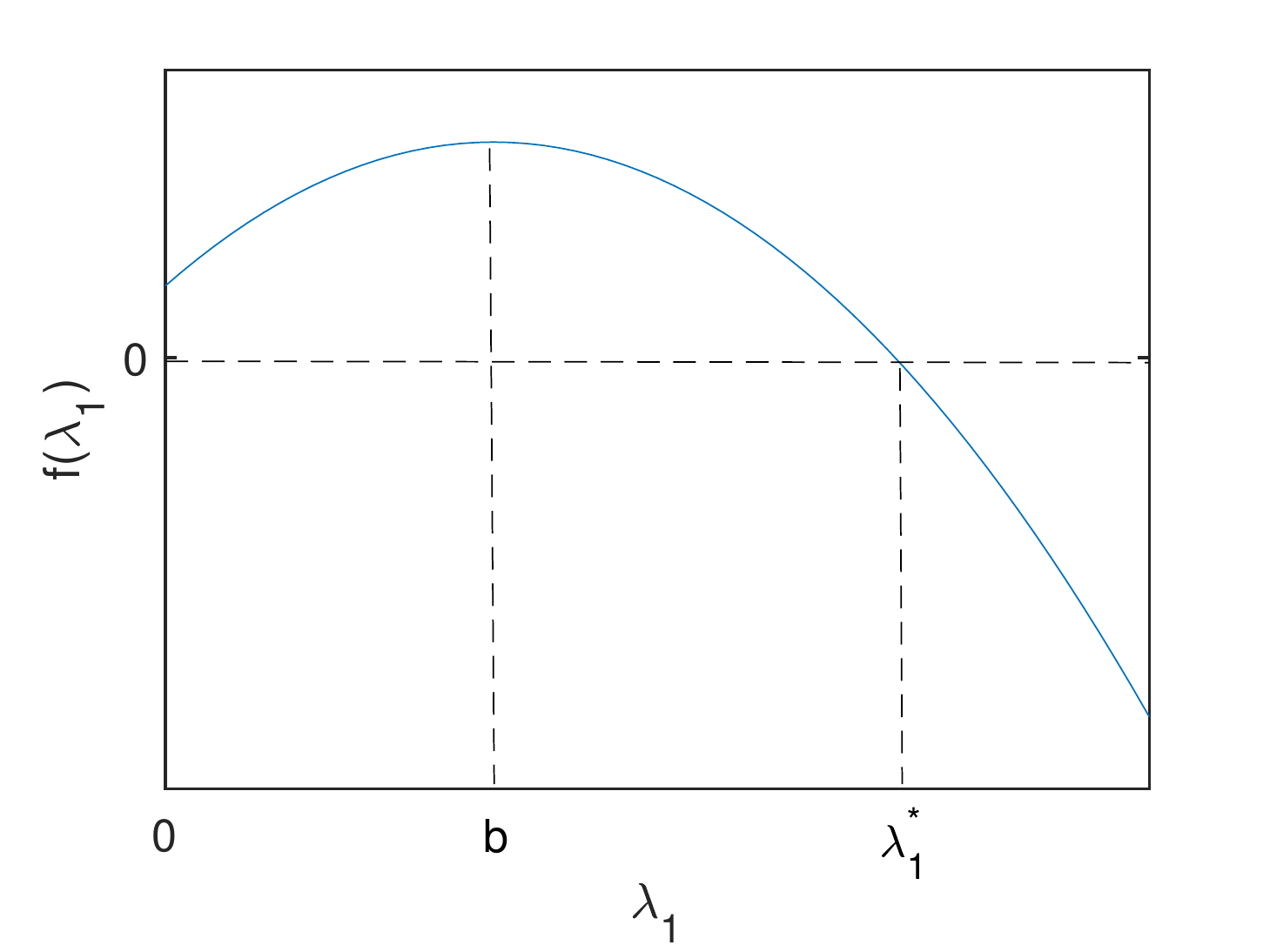}
\caption{$f(\lambda_1)$ versus $\lambda_1$ in Appendix A.B.}
 \end{figure}
 
 Suppose that when $N = M-1$, (3) has unique solutions. With induction method, we need to prove when $N = M$, (3) has unique solutions. Similar to (28)-(29), we can rewrite $\lambda_i$ as a function of $\lambda_j$ and the $\lambda_i$ is concave and strictly increasing in each $\lambda_j$, where $i\in\{1, \cdots, M\}$ and $j \ne i$. If we introduce $\lambda_M$ as in (28)-(29) into other $\lambda_i$ as in (28)-(29), where $i\in\{1, \cdots, M-1\}$, we have $\lambda_i$ is still concave and strictly increasing in $\lambda_j$, where $j \in \{1, \cdots, M-1\}$ and $j \ne i$. Since we know when $N = M-1$, (2) has unique solutions. Then after introducing $\lambda_M$ as in (28)-(29) into other $\lambda_i$ as in (28)-(29), the new $M-1$ equations also have unique solutions. Then we prove that when $N = M$, (3) has unique solutions.

Since $\lambda_1(\lambda_2)$ in (25) has an additional item $\frac{1}{\lambda_2\mu}$ in the denominator than $\lambda_1(\lambda_2)$ in (28), and $\lambda_2(\lambda_1)$ in (26) has an additional item $\frac{1}{\lambda_1\mu}$ in the denominator than $\lambda_2(\lambda_1)$ in (29), thus solutions to  (25) and (26), $(\lambda_1^{**}, \lambda_2^{**})$, are smaller than solutions to (28) and (29),  $(\lambda_1^{*}, \lambda_2^{*})$. Then we have $\lambda_i^{*} \geq \lambda_i^{**}$, $i = 1, 2$. For general $N \geq 2$, each $\lambda_i$ as in (25) has an additional item $\frac{1}{\mu}\big(\sum\limits_{j = 1}^N\frac{1}{\lambda_j}-\frac{1}{\lambda_i}\big)$ in the denominator than $\lambda_i$ as in (28), thus solutions to (25) and (26) are smaller than solutions to (28) and (29) when $N \geq 2$, we have $\lambda_i^* \geq \lambda_i^{**}$, $i\in\{1, \cdots, N\}$.

\section{Proof of Proposition \Romannum{2}.3}
\subsection{Proof of Social Optimizers and the Uniqueness}
To show (8)-(10) are solutions as social optimizers, note that $\pi((\lambda_1(c_H), \lambda_1(c_L)), \cdots, \lambda_N)$ is concave with each $\lambda_i$ due to $\frac{\partial^2 \pi((\lambda_1(c_H), \lambda_1(c_L)), \cdots, \lambda_N)}{\partial \lambda_i^2}\leq 0$, where $i\in\{1, \cdots, N\}$. By using the first-order condition, we have (8)-(10) as the solutions. Note that (8)-(10) have the same structure as (2), we can prove uniqueness of (8)-(10) by following proof of (2) in Appendix A.A. We thus skip details here.

\subsection{Proof of Competition Equilibrium and the Uniqueness}
To show (11)-(13) are solutions as equilibrium, note that each $\pi_i(\lambda_i, \lambda_{-i})$ is concave with each $\lambda_i$ due to $\frac{\partial^2 \pi_i(\lambda_i, \lambda_{-i})}{\partial \lambda_i^2}\leq 0$, where $i\in\{1, \cdots, N\}$. By using the first-order condition, we have (11)-(13) as the solutions. Note that (11)-(13) have the same structure as (3), we can prove uniqueness of (11)-(13) by following proof of (3) in Appendix A.B. We thus skip details here.

\section{Proof of Proposition \Romannum{2}.5}

If all the platforms have complete information of platform 1's sampling cost, they will play a one-shot game as in (1),
we present $((\bar{\lambda}_1(c_H),\bar{\lambda}_1(c_L))$, $(\bar{\lambda}_2(c_H),\bar{\lambda}_2(c_L)), \cdots, (\bar{\lambda}_N(c_H),\bar{\lambda}_N(c_L)))$, the Nash equilibrium, as unqiue solutions to
\vspace{-10pt}
\begin{table}[H]
\small{
   \begin{align}
    \bar{\lambda}_1(c_H) &= \sqrt{\frac{1+ \sum\limits_{j =2}^N \bar{\lambda}_j(c_H)/\mu}{c_H}}, \bar{\lambda}_i(c_H) = \sqrt{\frac{1+\big(\sum\limits_{j = 1}^N\bar{\lambda}_j(c_H)-\bar{\lambda}_i(c_H)\big)/\mu}{c_i}},\;\; i \in\{2, \cdots, N\}, \\
 \bar{\lambda}_1(c_L) &= \sqrt{\frac{1+ \sum\limits_{j =2}^N \bar{\lambda}_j(c_L)/\mu}{c_L}}, \bar{\lambda}_i(c_L) = \sqrt{\frac{1+\big(\sum\limits_{j = 1}^N\bar{\lambda}_j(c_L)-\bar{\lambda}_i(c_L)\big)/\mu}{c_i}}, \;\; i \in\{2, \cdots, N\}.
\end{align}}
\vspace{-30pt}
\end{table}
In the following, we first prove platform 1 obtains larger one-shot cost under incomplete information than that under complete information when $c_1=c_H$, then prove platform 1 can obtain larger one-shot average cost under incomplete information.

Platform 1's one-shot cost when $c_1 = c_H$ under incomplete information is $\pi_1(\lambda_1^*(c_H), \sum\limits_{j = 2}^N\lambda_j^*)$, where $\lambda_1^*(c_H)$ and $\lambda_j^*$ are given in (11)-(13), and its one-shot cost when $c_1 = c_H$ under complete information is $\pi_1(\bar{\lambda}_1(c_H), \sum\limits_{j=2}^N\bar{\lambda}_j(c_H))$ using (31).
To prove platform 1 obtains larger one-shot cost under incomplete information when $c_1 = c_H$, it's equivalent to prove $\pi_1(\lambda_1^*(c_H), \sum\limits_{j = 2}^N\lambda_j^*) \geq \pi_1(\bar{\lambda}_1(c_H), \sum\limits_{j=2}^N\bar{\lambda}_j(c_H))$. We want to prove the equivalent statement via introducing an intermediate term:
\begin{align}
    \pi_1(\lambda_1^*(c_H), \sum\limits_{j = 2}^N\lambda_j^*) \geq \pi_1(\lambda_1^*(c_H), \sum\limits_{j=2}^N\bar{\lambda}_j(c_H)) \geq \pi_1(\bar{\lambda}_1(c_H), \sum\limits_{j=2}^N\bar{\lambda}_j(c_H)). 
\end{align}
To prove (33), we first prove $\pi_1(\lambda_1^*(c_H), \sum\limits_{j = 2}^N\lambda_j^*) \geq \pi_1(\lambda_1^*(c_H), \sum\limits_{j=2}^N\bar{\lambda}_j(c_H))$. By comparing (11)-(13) and (31),  
we have $ \sum\limits_{j = 2}^N\lambda_j^* \geq \sum\limits_{j=2}^N\bar{\lambda}_j(c_H)$. Since $\pi_1(\lambda_1, \sum\limits_{j =2}^N\lambda_j)$ in (1) increases with $\sum\limits_{j =2}^N\lambda_j$, we have $\pi_1(\lambda_1^*(c_H), \sum\limits_{j = 2}^N\lambda_j^*) \geq \pi_1(\lambda_1^*(c_H), \sum\limits_{j=2}^N\bar{\lambda}_j(c_H))$.

We then prove $\pi_1(\lambda_1^*(c_H), \sum\limits_{j=2}^N\bar{\lambda}_j(c_H)) \geq\pi_1(\bar{\lambda}_1(c_H), \sum\limits_{j=2}^N\bar{\lambda}_j(c_H))$ in (33). According to (31), the best response $\lambda_1 = \bar{\lambda}_1(c_H)$ is the unique solution to minimize platform 1's one-shot cost $\pi_1(\lambda_1, \sum\limits_{j =2}^N\bar{\lambda}_j(c_H))$, thus $\pi_1(\lambda_1^*(c_H), \sum\limits_{j=2}^N\bar{\lambda}_j(c_H)) \geq\pi_1(\bar{\lambda}_1(c_H), \sum\limits_{j=2}^N\bar{\lambda}_j(c_H))$. 


Similar to (33), we can prove 
\begin{align}
    \pi_1(\lambda_1^*(c_L), \sum\limits_{j = 2}^N\lambda_j^*)  \leq \pi_1(\bar{\lambda}_1(c_L), \sum\limits_{j = 2}^N\lambda_j^*)
  \leq \pi_1(\bar{\lambda}_1(c_L), \sum\limits_{j=2}^N\bar{\lambda}_j(c_L)).
\end{align}

Finally, we are ready to prove that platform 1 can obtain larger one-shot average cost under incomplete information than that under complete information. This time-average cost of platform 1 under complete information is
\begin{align}
   \pi_1^C = p_H\pi_1(\bar{\lambda}_1(c_H), \sum\limits_{j=2}^N\bar{\lambda}_j(c_H)) + (1-p_H)\pi_1(\bar{\lambda}_1(c_L), \sum\limits_{j=2}^N\bar{\lambda}_j(c_L)). \nonumber
\end{align}
Its average cost under information advantage changes to 

\begin{align}
    \pi_1^I = p_H\pi_1(\lambda_1^*(c_H), \sum\limits_{j = 2}^N\lambda_j^*) + (1-p_H)\pi_1(\lambda_1^*(c_L), \sum\limits_{j = 2}^N\lambda_j^*), \nonumber
\end{align}
where $\lambda_1^*(c_H)$, $\lambda_1^*(c_L)$ and  $\lambda_j^*$ are given in (11)-(13). If platform 1 obtains larger cost under incomplete information, or $\pi_1^C \leq \pi_1^I$, which can be simplified as
\vspace{-10pt}
\begin{table}[H]
   \small{
   \begin{align}
p_H&\geq \frac{\pi_1(\bar{\lambda}_1(c_L),  \sum\limits_{j=2}^N\bar{\lambda}_j(c_L)) -  \pi_1(\lambda_1^*(c_L), \sum\limits_{j = 2}^N\lambda_j^*)}{\pi_1(\bar{\lambda}_1(c_L), \sum\limits_{j=2}^N\bar{\lambda}_j(c_L)) -  \pi_1(\lambda_1^*(c_L), \sum\limits_{j = 2}^N\lambda_j^*) + \pi_1(\lambda_1^*(c_H), \sum\limits_{j = 2}^N\lambda_j^*) - \pi_1(\bar{\lambda}_1(c_H),  \sum\limits_{j=2}^N\bar{\lambda}_j(c_H))}, \nonumber 
\end{align}
   }
\vspace{-20pt}
\end{table}
\noindent in which the right-hand side is positive and less than 1 because of (33) and (34). 
\section{Proof of Proposition \Romannum{3}.3}
To show the equation in Proposition \Romannum{3}.3 are the solutions to the cooperation profile, we solve and rewrite (17) as
\begin{align}
    \tilde{\lambda}_i^2(\delta)/c_i - 2\sqrt{\delta\lambda_i^* + (1-\delta)\sqrt{\frac{1+\frac{\lambda_{-i}(\delta)}{\mu}}{c_i}}}\tilde{\lambda}_i(\delta) + 1+\frac{\lambda_{-i}(\delta)}{\mu} = 0, \nonumber
\end{align}
where $\lambda_{-i}(\delta) = \sum\limits_{k =1 }^{j}\tilde{\lambda}_k(\delta) - \tilde{\lambda}_i(\delta) + \sum\limits_{k = j+1}^N\lambda_k^{**}$. We then choose to take the smaller root with smaller social cost, which is consistent with equation in Proposition \Romannum{3}.3.

Notice that the only different between the equation in Proposition \Romannum{3}.3. and (19) is that in latter, $\tilde{\lambda}_k(\delta) = \lambda_k^{**}$ are constant for $k \in \{j+1, \cdots, N\}$, while in (19), such $\tilde{\lambda}_k(\delta)$s are still variable to determine. Then we can prove unique solution of the equation in Proposition \Romannum{3}.3. by proving that of (19), which is given in Appendix E.

We rewrite the equation in Proposition \Romannum{3}.3 as 
\begin{align}
    \tilde{\lambda}_i(\delta) = \delta\lambda_i^* + (1-\delta)\sqrt{\frac{1+\frac{\lambda_{-i}(\delta)}{\mu}}{c_i}} 
   -\sqrt{\bigg(\delta\lambda_i^* + (1-\delta)\sqrt{\frac{1+\frac{\lambda_{-i}(\delta)}{\mu}}{c_i}} \bigg)^2 - \frac{1+\frac{\lambda_{-i}(\delta)}{\mu}}{c_i}}.
\end{align}
$\tilde{\lambda}_i(\delta)$ in (35) decreases with $\delta$ due to $\frac{\partial \tilde{\lambda}_i(\delta)}{\partial \delta} < 0$, where $i \in\{1, \cdots, j\}$. When $\delta = 0$,  $\tilde{\lambda}_i(\delta) = \lambda_i^*$ in (35). When $\delta = \delta^{th}_i$,  $\tilde{\lambda}_i(\delta) = \lambda_i^{**}$ in (35). Thus we prove that $\lambda_i^{**} < \tilde{\lambda}_i(\delta) < \lambda_i^{*}$ and $\tilde{\lambda}_i(\delta)$ decreases with $\delta$.

\section{Proof for Proposition \Romannum{3}.4}
To show (19) are the solutions to the cooperation profile, we solve and rewrite (18) as
\begin{align}
    \tilde{\lambda}_i^2(\delta)/c_i - 2\sqrt{\delta\lambda_i^* + (1-\delta)\sqrt{\frac{1+\frac{\lambda_{-i}(\delta)}{\mu}}{c_i}}}\tilde{\lambda}_i(\delta) + 1+\frac{\lambda_{-i}(\delta)}{\mu} = 0, \nonumber
\end{align}
where $\lambda_{-i}(\delta) = \sum\limits_{k =1 }^{N}\tilde{\lambda}_k(\delta) - \tilde{\lambda}_i(\delta)$. We then choose to take the smaller root with smaller social cost, which is consistent with (19).

We want to show (19) have unique solutions with induction method. When $N =2$, we rewrite (18) as
\begin{align}
    \tilde{\lambda}_{1}^2(\delta) - 2\bigg(\delta\lambda_1^* + (1-\delta)\sqrt{\frac{1+\tilde{\lambda}_{2}(\delta)/\mu}{c_1}}\bigg)\tilde{\lambda}_{1}(\delta) + \frac{1+\tilde{\lambda}_{2}(\delta)/\mu}{c_1} &= 0, \\
    \tilde{\lambda}_{2}^2(\delta) - 2\bigg(\delta\lambda_2^* + (1-\delta)\sqrt{\frac{1+\tilde{\lambda}_{1}(\delta)/\mu}{c_2}}\bigg)\tilde{\lambda}_{2}(\delta) + \frac{1+\tilde{\lambda}_{1}(\delta)/\mu}{c_2} &= 0. 
\end{align}
Denote
\begin{align}
    f(\tilde{\lambda}_{1}(\delta)) = &\tilde{\lambda}_{1}^2(\delta) - 2\bigg(\delta\lambda_1^* + (1-\delta)\sqrt{\frac{1+\tilde{\lambda}_{2}(\delta)/\mu}{c_1}}\bigg)\tilde{\lambda}_{1}(\delta) + \frac{1+\tilde{\lambda}_{2}(\delta)/\mu}{c_1}. \nonumber
\end{align}
In the range of $\lambda_2^{**}<\tilde{\lambda}_2(\delta)<\lambda_2^{*}$, symmetric axis of $ f(\tilde{\lambda}_{1}(\delta))$ satisfies 
\begin{align}
   \lambda_1^{**} < \delta\lambda_1^* + (1-\delta)\sqrt{\frac{1+\tilde{\lambda}_{2}(\delta)/\mu}{c_1}} < \lambda_1^*.\nonumber
\end{align}
To show $f(\lambda_1^{**}) > 0$, we simplify it as

 \begin{align}
    \delta < \delta_1 = \frac{\bigg( \lambda_1^{**} - \sqrt{\frac{1+\tilde{\lambda}_2(\delta)/\mu}{c_1}} \bigg)^2}{2\lambda_1^{**}\bigg( \lambda_1^* - \sqrt{\frac{1+\tilde{\lambda}_2(\delta)/\mu}{c_1}} \bigg)}. \nonumber
\end{align}
Since $\delta_1$ increases with $\tilde{\lambda}_2(\delta) \in (\lambda_2^{**}, \lambda_2^{*})$ and when $\tilde{\lambda}_2(\delta) = \lambda_2^{**}$, $\delta_1 = \delta_{th1}$, then $\delta_1 > \delta_{th1}$ always holds in tits range. Then $\delta < \delta_1$ always holds because $\delta < \delta_{th1}$, which means $f(\lambda_1^{**}) > 0$ holds. Then the smaller positive root of (36) must be in the range of  $\lambda_1^{**}<\tilde{\lambda}_1(\delta)<\lambda_1^{*}$. Similarly, we can show that the smaller positive root of (37) is also in the range of  $\lambda_2^{**}<\tilde{\lambda}_2(\delta)<\lambda_2^{*}$.

Notice that although (36) and (37) may have multiple roots, we only select the root that has smallest social cost, which are the solutions to
\begin{align}
    \tilde{\lambda}_1(\delta)(\tilde{\lambda}_2(\delta)) &= \delta\lambda_1^* + (1-\delta)\sqrt{\frac{1+\frac{\tilde{\lambda}_2(\delta)}{\mu}}{c_1}} - \sqrt{\bigg(\delta\lambda_1^* + (1-\delta)\sqrt{\frac{1+\frac{\tilde{\lambda}_2(\delta)}{\mu}}{c_1}} \bigg)^2 - \frac{1+\frac{\tilde{\lambda}_2(\delta)}{\mu}}{c_1}},  \\
   \tilde{\lambda}_2(\delta)(\tilde{\lambda}_1(\delta)) &= \delta\lambda_2^* + (1-\delta)\sqrt{\frac{1+\frac{\tilde{\lambda}_1(\delta)}{\mu}}{c_2}} - \sqrt{\bigg(\delta\lambda_2^* + (1-\delta)\sqrt{\frac{1+\frac{\tilde{\lambda}_1(\delta)}{\mu}}{c_2}} \bigg)^2 - \frac{1+\frac{\tilde{\lambda}_1(\delta)}{\mu}}{c_2}}.
\end{align}
We've shown the existence of (38) and (39) in previous paragraph.
To show (38)-(39) have unique solutions in the range of $\lambda_1^{**}<\tilde{\lambda}_1(\delta)<\lambda_1^{*}, \lambda_2^{**}<\tilde{\lambda}_2(\delta)<\lambda_2^{*}$, denote 
\begin{align}
    g(\tilde{\lambda}_1(\delta)) = \tilde{\lambda}_2(\delta)(\tilde{\lambda}_1(\delta)) - \tilde{\lambda}_2(\delta)_1(\tilde{\lambda}_1(\delta)), \nonumber
\end{align}
 where $ \tilde{\lambda}_2(\delta)_1(\tilde{\lambda}_1(\delta))$ is the inverse function to $\tilde{\lambda}_1(\delta)(\tilde{\lambda}_2(\delta))$ in (38) with variable $\tilde{\lambda}_1(\delta)$.
By taking first and second derivatives of $\tilde{\lambda}_1(\delta)(\tilde{\lambda}_2(\delta))$ in (38) and of $\tilde{\lambda}_2(\delta)(\tilde{\lambda}_1(\delta))$ in (39), we can find that $\tilde{\lambda}_1(\delta)(\tilde{\lambda}_2(\delta))$ is convex and strictly increasing in $\tilde{\lambda}_2(\delta)$, and  $\tilde{\lambda}_2(\delta)(\tilde{\lambda}_1(\delta))$ is convex and strictly increasing in $\tilde{\lambda}_1(\delta)$. Thus $g(\tilde{\lambda}_1(\delta))$ is convex in $\tilde{\lambda}_1(\delta)$. Additionally, we have
\begin{align}
    g(\lambda_1^{**}) &> 0, \nonumber \\
    g(\lambda_1^{*}) &= 0, \nonumber \\
    g'(\lambda_1^*) &> 0. \nonumber
\end{align}
Thus there exists unique $\tilde{\lambda}_1^0(\delta)$ in $\lambda_1^{**} < \tilde{\lambda}_1(\delta) < \lambda_1^*$ satisfying $g(\tilde{\lambda}_1^0(\delta)) = 0$. We plot $g(\tilde{\lambda}_1(\delta))$ in Figure 8. Then 
there exist unique solutions to (38) and (39) in the feasible range of $\lambda_1^{**} < \tilde{\lambda}_1(\delta) < \lambda_1^*$ and $\lambda_2^{**} < \tilde{\lambda}_1(\delta) < \lambda_2^*$.
\begin{figure}
\centering
 \includegraphics[height=2.5in, width=3in]{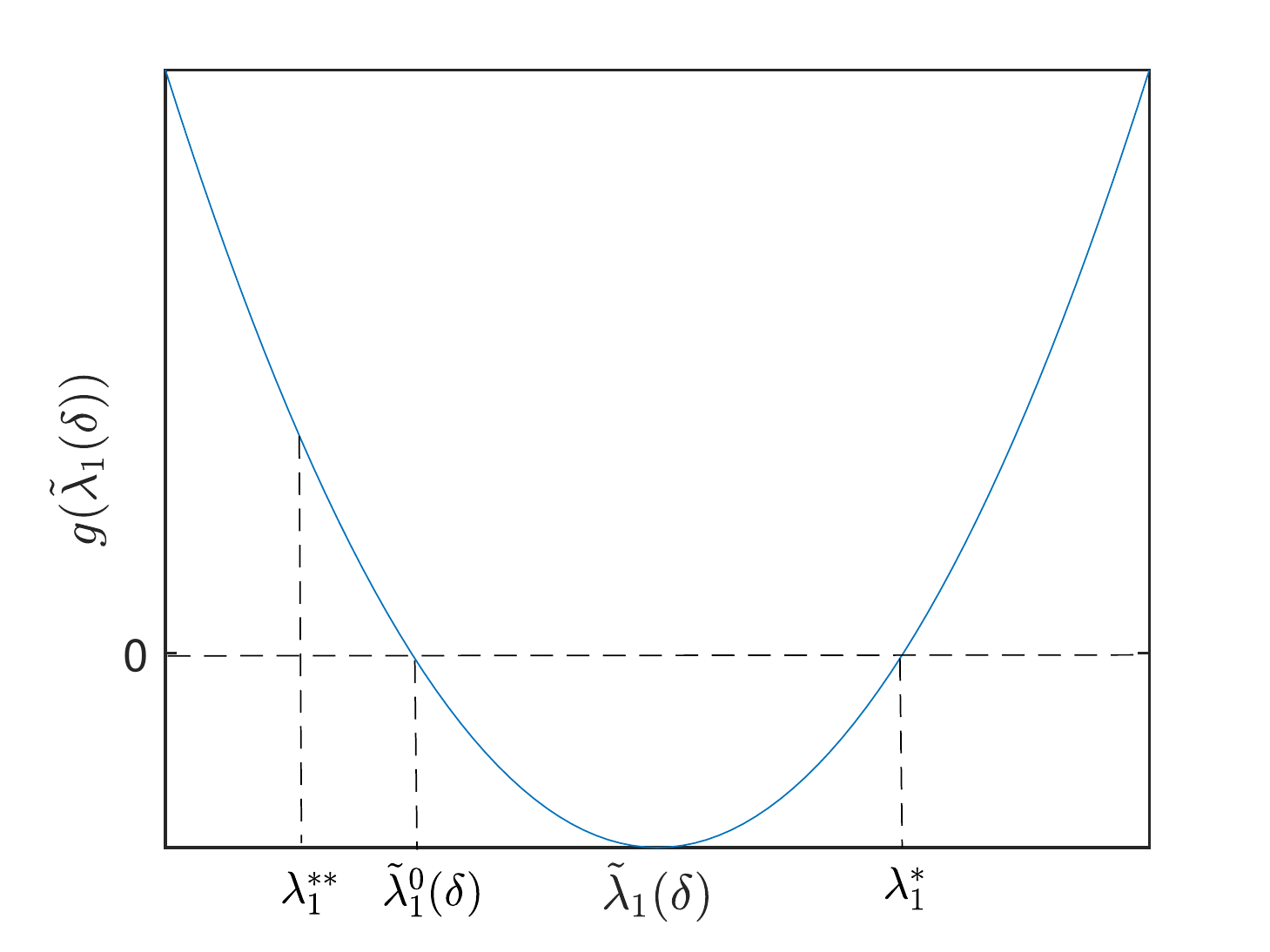}
\caption{$ g(\tilde{\lambda}_1(\delta))$ versus $\tilde{\lambda}_1(\delta)$ in Appendix E.}
 \end{figure}
 
 Suppose that when $N = M-1$, (19) has unique solutions. With induction method, we need to prove when $N = M$, (19) has unique solutions. Similar to (38)-(39), we can rewrite $\tilde{\lambda}_i(\delta)$ as a function of $\tilde{\lambda}_j(\delta)$ and the $\tilde{\lambda}_i(\delta)$ is convex and strictly increasing in each $\tilde{\lambda}_j(\delta)$, where $i\in\{1, \cdots, M\}$ and $j \ne i$. If we introduce $\tilde{\lambda}_M(\delta)$ as in (38)-(39) into other $\tilde{\lambda}_i(\delta)$ as in (38)-(39), where $i\in\{1, \cdots, M-1\}$, we have $\tilde{\lambda}_i(\delta)$ is still convex and strictly increasing in $\tilde{\lambda}_j(\delta)$, where $j \in \{1, \cdots, M-1\}$ and $j \ne i$. Since we know when $N = M-1$, (19) has unique solutions. Then after introducing $\lambda_M$ as in (38)-(39) into other $\lambda_i$ as in (38)-(39), the new $M-1$ equations also have unique solutions. Then we prove that when $N = M$, (19) has unique solutions.

When $\delta \to 0$, (19) just becomes (2), then solution to (19) are $ \tilde{\lambda}_{i}(\delta) = \lambda_i^*$, where $i\in\{1, \cdots, N\}$.
 
$\tilde{\lambda}_i(\delta)$ in (19) decreases with $\delta$ due to $\frac{\partial \tilde{\lambda}_i(\delta)}{\partial \delta} < 0$, where $i \in\{1, \cdots, M\}$. Thus solutions to (19) decrease with $\delta$.

\section{Proof of Lemma \Romannum{4}.1}

Given $c_1 = c_H$, at the social optimum platform 1's sampling rate is $\lambda_1 = \lambda_1^{**}(c_H)$ with corresponding cost: 
$$\pi_1(c_H, \lambda_1^{**}(c_H)) = \frac{1}{\lambda_1^{**}(c_H)}\bigg(1+\big(\sum\limits_{i=2}^N\lambda_i^{**}\big)/\mu\bigg) + c_H\lambda_1^{**}(c_H) + 1/\mu.$$  
If it cheats to play $\lambda_1^{**}(c_L)$, its cost will be $\pi_1(c_H, \lambda_1^{**}(c_L))$. 
If $\pi_1(c_H, \lambda_1^{**}(c_H)) \leq  \pi_1(c_H, \lambda_1^{**}(c_L))$, platform 1 won't deviate to $\lambda_1 = \lambda_1^{**}(c_L)$, which is equivalent to
\vspace{-10pt}
\begin{table}[H]
    \small{
    $$\sqrt{c_H + \frac{1}{\mu}\sum\limits_{i=2}^N\frac{1}{\lambda_i^{**}}}\sqrt{c_L + \frac{1}{\mu}\sum\limits_{i=2}^N\frac{1}{\lambda_i^{**}}} \leq c_H,$$
    }
    \vspace{-30pt}
\end{table}
\noindent which holds only if $\frac{1}{\mu}\sum\limits_{i=2}^N\frac{1}{\lambda_i^{**}}$ is small and is not generally true. Then platform 1 may deviate from $\lambda_1^{**}(c_H)$ to $\lambda_1^{**}(c_L)$ when $c_1 = c_H$.

\section{Proof of Proposition \Romannum{4}.5 }
Before we prove the proposition, let's first prove an useful lemma. 
\begin{lemma}
Given real numbers $x, y, a, b > 0$ and $x \geq y$, $b \geq a$, $\frac{x+a}{y+b} \leq \frac{x}{y}$. 
\end{lemma}
\begin{proof}
Since $x, y, a, b$ are all positive, $\frac{x+a}{y+b} \leq \frac{x}{y}$ is equivalent to $bx \geq ay$. Then we continue to prove $bx \geq ay$ is true, which holds 
due to given condition $x \geq y > 0$, $b \geq a > 0$.
\end{proof}

When $p_Hc_H + (1-p_H)c_L = c_{2} = \cdots = c_N$, according to (20)-(21), the approximate cooperation profile in larger $\delta$ regime are 
\begin{align}
    \hat{\lambda}_1 = \hat{\lambda}_2 = \cdots = \hat{\lambda}_N = \frac{1}{\sqrt{p_Hc_H + (1-p_H)c_L}}. \nonumber
\end{align}
As the platforms repeat their sampling choices in the social optimum and the mechanism, we only need to compare their one-shot social costs. The social cost in one-shot under this approximate cooperation profile is
\begin{align}
    \pi^{r} = 2N\sqrt{p_Hc_H + (1-p_H)c_L} + N^2/\mu. \nonumber
\end{align}
Since $c_2 = \cdots = c_N$, according to (8)-(10), we have $\lambda_2^{**}=\cdots=\lambda_N^{**}$. Then the minimum social cost in one-shot can be simplified as 
\begin{align}
    \pi^{**} 
    = &p_H\bigg(\frac{1}{\lambda_1^{**}(c_H)}+c_H\lambda_1^{**}(c_H)\bigg) + (1-p_H)\bigg(\frac{1}{\lambda_1^{**}(c_L)}+c_L\lambda_1^{**}(c_L)\bigg)+ N/\mu \nonumber \\
    &+ (N-1)\bigg(\frac{1}{\lambda_2^{**}}+c_2\lambda_2^{**}\bigg) + (N-2)(N-1)/\mu \nonumber \\
    &+(N-1)p_H\bigg(\frac{\lambda_2^{**}}{\lambda_1^{**}(c_H)} + \frac{\lambda_1^{**}(c_H)}{\lambda_2^{**}} \bigg) +(N-1)(1-p_H)\bigg(\frac{\lambda_2^{**}}{\lambda_1^{**}(c_L)} + \frac{\lambda_1^{**}(c_L)}{\lambda_2^{**}} \bigg). \nonumber
\end{align}
Then the ratio of social costs under the approximate cooperation profile and social optimizers is 
\begin{align}
    r = \max_{0\leq p_H\leq 1, \mu, c_H > c_L > 0} \frac{\pi^{r}}{\pi^{**}}. 
\end{align}
Since both $\pi^{r}$ and $\pi^{**}$ contain a common term $(N+(N-2)(N-1))/\mu$, from Lemma C.1 we know if we eliminate this term from both $\pi^{r}$ and $\pi^{**}$, the ratio would just become larger. Thus, we can rewrite (40) as
\begin{align}
   r \leq \max_{0\leq p_H\leq 1, \mu, c_H > c_L > 0}\frac{2N\sqrt{p_Hc_H + (1-p_H)c_L} +\frac{2N-2}{\mu}}{m_1 + m_2}, 
\end{align}
where all the $N-1$ platforms behave the same, and 
\vspace{-20pt}
\begin{table}[H]
    \small{
    \begin{align}
   m_1 &= p_H\bigg(\frac{1}{\lambda_1^{**}(c_H)}+c_H\lambda_1^{**}(c_H)\bigg) + (1-p_H)\bigg(\frac{1}{\lambda_1^{**}(c_L)}+c_L\lambda_1^{**}(c_L)\bigg)+ (N-1)\bigg(\frac{1}{\lambda_2^{**}}+c_2\lambda_2^{**}\bigg), \nonumber \\
    m_2  &= 
     (N-1)p_H\bigg(\frac{\lambda_2^{**}}{\lambda_1^{**}(c_H)} + \frac{\lambda_1^{**}(c_H)}{\lambda_2^{**}} \bigg) +(N-1)(1-p_H)\bigg(\frac{\lambda_2^{**}}{\lambda_1^{**}(c_L)} + \frac{\lambda_1^{**}(c_L)}{\lambda_2^{**}} \bigg). \nonumber  
\end{align}
    }
    \vspace{-30pt}
\end{table}
Given $p_H = 0$ or $p_H = 1$, $\pi^{r} = \pi^{**}$ always holds. If $0 < p_H < 1$, we notice that in $m_2$, $\frac{\lambda_2^{**}}{\lambda_1^{**}(c_H)} + \frac{\lambda_1^{**}(c_H)}{\lambda_2^{**}} \geq 2$ and $\frac{\lambda_2^{**}}{\lambda_1^{**}(c_L)} + \frac{\lambda_1^{**}(c_L)}{\lambda_2^{**}} \geq 2$, where the equalities hold at $\lambda_2^{**} = \lambda_1^{**}(c_H) = \lambda_1^{**}(c_L)$. Then we can tell that $m_2$ is minimized at $\lambda_2^{**} = \lambda_1^{**}(c_H) = \lambda_1^{**}(c_L)$ with $m_2 = (2N-2) / \mu$. By using Lemma C.1 again, we eliminate $(2N-2)/\mu$ in the numerator of (41) and $m_2$ in the denominator and rewrite (41) as
\begin{align}
    r \leq \max_{0<p_H<1, \mu, c_H>c_L>0}\frac{2N\sqrt{p_Hc_H + (1-p_H)c_L}}{m_1}. 
\end{align}
Now we focus on $m_1$. Inside, $\frac{1}{\lambda_1^{**}(c_H)}+c_H\lambda_1^{**}(c_H)$ is minimized at $\lambda_1^{**}(c_H) = \frac{1}{\sqrt{c_H}}$, $\frac{1}{\lambda_1^{**}(c_L)}+c_L\lambda_1^{**}(c_L)$ is minimized at $\lambda_1^{**}(c_L) = \frac{1}{\sqrt{c_L}}$ and $\frac{1}{\lambda_2^{**}}+c_2\lambda_2^{**}$ is minimized at $\lambda_2^{**} = \frac{1}{\sqrt{c_2}}$. Thus, $m_1$ is minimized at $\lambda_1^{**}(c_H) = \frac{1}{\sqrt{c_H}}$, $\lambda_1^{**}(c_L) = \frac{1}{\sqrt{c_L}}$ and $\lambda_2^{**} = \frac{1}{\sqrt{c_2}}=\frac{1}{\sqrt{p_Hc_H+(1-p_H)c_L}}$, which only happen altogether at $\mu \to \infty$. Thus, the right-hand side of (42) is maximized at $\mu \to \infty$. We thus simply (42) as
\begin{align}
    r &\leq \max_{0<p_H<1, c_H>c_L>0}\frac{N\sqrt{p_Hc_H/c_L + (1-p_H)}}{p_H\sqrt{c_H/c_L} + (1-p_H) + (N-1)\sqrt{p_H c_H/c_L+(1-p_H)}}.
\end{align}
If we replace $c_H$, $c_L$ by $x = \sqrt{c_H/c_L} $, (43) is simplified to
\begin{align}
    r  \leq \max_{0<p_H<1, x > 1}\frac{N\sqrt{p_Hx^2 + (1-p_H)}}{p_Hx + (1-p_H) + (N-1)\sqrt{p_H x^2+(1-p_H)}}:=f(x).
\end{align}
We notice that $f(x)$ increases with $x > 1$ because 
\begin{align}
    f'(x) = \frac{N(1-p_H)p_H(x-1)}{\bigg(p_Hx + (1-p_H) + (N-1)\sqrt{p_H x^2+(1-p_H)}\bigg)^2\sqrt{p_Hx^2+1-p_H}} > 0. \nonumber
\end{align}
Thus, the right-hand side of (44) is maximized at $x \to \infty$ and we can finally rewrite (44) as
\vspace{-10pt}
\begin{table}[H]
    \begin{align}
    r  &\leq \lim_{x\to\infty}\max_{0<p_H<1}\frac{N\sqrt{p_Hx^2 + (1-p_H)}}{p_H x + (1-p_H) + (N-1) \sqrt{p_H x^2+(1-p_H)}} \nonumber \\
         &= \max_{0<p_H<1}\frac{N\sqrt{p_H}}{p_H + (N-1)\sqrt{p_H}} \nonumber = \max_{0<p_H<1}\frac{N}{\sqrt{p_H} + N-1} < \frac{N}{N-1}. \nonumber
\end{align}
\vspace{-20pt}
\end{table}

\section{Proof of Proposition \Romannum{4}.6}
To show the equations in Proposition \Romannum{4}.6 are the solutions, we need to solve the following equations for platform 1
\begin{align}
\delta=\max \left\{\hat{\delta}_{1}^{t h}\left(\tilde{\lambda}_{1}(\delta), \tilde{\lambda}_{2}(\delta), \cdots, \tilde{\lambda}_{N}(\delta) | c_{1}=c_{L}\right), \hat{\delta}_{1}^{t h}\left(\tilde{\lambda}_{1}(\delta), \tilde{\lambda}_{2}(\delta), \cdots, \tilde{\lambda}_{N}(\delta) | c_{1}=c_{H}\right)\right\}.
\end{align}
 and for platform $i\in\{2, \cdots, N\}$:  
\begin{align}
  \delta=\frac{\pi_{i}\left(\tilde{\lambda}_{i}(\delta), \tilde{\lambda}_{-i}(\delta)\right)-\pi_{i}\left(\tilde{\lambda}_{-i}(\delta), \sqrt{\frac{1+\tilde{\lambda}_{-i}(\delta) / \mu}{c_{i}}}\right)}{\pi_{i}\left(\lambda_{i}^{*}, \lambda_{-i}^{*}\right)-\pi_{i}\left(\tilde{\lambda}_{-i}(\delta), \sqrt{\frac{1+\tilde{\lambda}_{-i}(\delta) / \mu}{c_{i}}}\right)}.
\end{align}
where $\tilde{\lambda}_{-i}(\delta)=\sum\limits_{j \neq i}^{N} \tilde{\lambda}_{j}(\delta) \text { and } \lambda_{-i}^{*}=\sum\limits_{m \neq i, 1}^{N} \lambda_{m}^{*}+p_{H} \lambda_{1}^{*}\left(c_{H}\right)+\left(1-p_{H}\right) \lambda_{1}^{*}\left(c_{L}\right)$.
By jointly solving (45) and (46) and taking the smaller roots to avoid large social cost, we have the following cooperation $\big(\tilde{\lambda}_1(\delta), \tilde{\lambda}_2(\delta), \cdots, \tilde{\lambda}_N(\delta)\big)$ for all the $N$ platforms as in Proposition \Romannum{4}.6.

Notice that the only different between the equations in Proposition \Romannum{4}.6 and Proposition \Romannum{4}.7 is that in the equations in Proposition \Romannum{4}.6, $\tilde{\lambda}_k(\delta) = \hat{\lambda}_k$ are constant for $k \in \{j, \cdots, N\}$, while in the equations in Proposition \Romannum{4}.7, such $\tilde{\lambda}_k(\delta)$s are still variable to determine. Then we can prove unique solution of the equations in Proposition \Romannum{4}.6 by proving that of the equations in Proposition \Romannum{4}.7, which is given in Appendix I.

\section{Proof of Proposition \Romannum{4}.7}

Platform 1 should not deviate with $\tilde{\lambda}_1(\delta)$ whether $c_1 = c_H$ or $c_1 = c_L$ and platform $i = 2, 3, \cdots, N$ should not deviate with $\tilde{\lambda}_i(\delta)$, which are equivalent to
\begin{align}
\delta \geq \hat{\delta}^{th}_1(\tilde{\lambda}_1(\delta),\sum\limits_{j=2}^N\tilde{\lambda}_j(\delta) | c_1=c_L),  
\delta \geq \hat{\delta}^{th}_1(\tilde{\lambda}_1(\delta),\sum\limits_{j=2}^N\tilde{\lambda}_j| c_1=c_H), 
\delta \geq \hat{\delta}^{th}_i(\sum\limits_{j \ne i}^N\tilde{\lambda}_j(\delta)).
\end{align}
Solutions to (47) are
\begin{table}[H]
\small{
\begin{align}     
 \tilde{\lambda}_1(\delta) &\in \bigg[\frac{M_L' - \sqrt{M_L^{2'} -(\delta\hat{c}_1+(1-\delta)c_L)(1+\sum\limits_{j=2}^N\tilde{\lambda}_j(\delta)/\mu) }}{\delta\hat{c}_1+(1-\delta)c_L}, \frac{M_L' + \sqrt{M_L^{'2} -(\delta\hat{c}_1+(1-\delta)c_L)(1+\sum\limits_{j=2}^N\tilde{\lambda}_j(\delta)/\mu) }}{\delta\hat{c}_1+(1-\delta)c_L}\bigg], 
 \end{align}
 \begin{align} 
 \tilde{\lambda}_1(\delta) &\in \bigg[\frac{M_H' - \sqrt{M_H^{'2} -(\delta\hat{c}_1+(1-\delta)c_H)(1+\sum\limits_{j=2}^N\tilde{\lambda}_j(\delta)/\mu) }}{\delta\hat{c}_1+(1-\delta)c_H}, \nonumber \frac{M_H' + \sqrt{M_H^{'2} -(\delta\hat{c}_1+(1-\delta)c_H)(1+\sum\limits_{j=2}^N\tilde{\lambda}_j(\delta)/\mu) }}{\delta\hat{c}_1+(1-\delta)c_H}\bigg], 
 \end{align}
 \begin{align} 
 \tilde{\lambda}_i(\delta) & \in \bigg[\frac{M_i - \sqrt{M_i^{2} - c_i(1+ \sum\limits_{j\ne i}^N\tilde{\lambda}_j(\delta)/\mu)}}{c_i}, \frac{M_i + \sqrt{M_i^{2} - c_i(1+ \tilde{\lambda}_1(\delta)/\mu)}}{c_i} \bigg],
\end{align}
}
\end{table}
where 
\begin{table}[H]
\small{
\begin{align}
M_k' = \delta\sqrt{1+\sum\limits_{j=2}^N\lambda_j^*/\mu}(p_H\sqrt{c_H} + (1-p_H)\sqrt{c_k})+(1-\delta)\sqrt{(1+ \sum\limits_{j\ne i}^N\tilde{\lambda}_j(\delta))c_k }, \;\; k \in \{H, L\},  \nonumber
\end{align}
}

\end{table}
\begin{table}[H]
\small{
\begin{align}
M_i &= \sqrt{c_i}\bigg(\delta\sqrt{1+(p_H\lambda_1^*(c_H)+(1-p_H)\lambda_1^*(c_L)+\sum\limits_{j\ne i,1}^N\lambda_j^*)/\mu}
    +(1-\delta)\sqrt{1+ \sum\limits_{j\ne i}^N\tilde{\lambda}_j(\delta)/\mu}\bigg), \;\; i = 2,3,\cdots,N. \nonumber
\end{align}
}

\end{table}
Interaction of (48) and (49) is the feasible region for desired $\tilde{\lambda}_1(\delta)$, which is
\begin{table}[H]
\small{
\begin{align}     
 \tilde{\lambda}_1(\delta) \in \bigg[ &\max\bigg\{\frac{M_L' - \sqrt{M_L^{'2} -(\delta\hat{c}_1+(1-\delta)c_L)(1+\tilde{\lambda}_2(\delta)/\mu) }}{\delta\hat{c}_1+(1-\delta)c_L}, \frac{M_H' - \sqrt{M_H^{'2} -(\delta\hat{c}_1+(1-\delta)c_H)(1+\tilde{\lambda}_2(\delta)/\mu) }}{\delta\hat{c}_1+(1-\delta)c_H} \bigg\}, \nonumber \\
 &\min\bigg\{\frac{M_L' + \sqrt{M_L^{'2} -(\delta\hat{c}_1+(1-\delta)c_L)(1+\tilde{\lambda}_2(\delta)/\mu) }}{\delta\hat{c}_1+(1-\delta)c_L}, \frac{M_H' + \sqrt{M_H^{'2} -(\delta\hat{c}_1+(1-\delta)c_H)(1+\tilde{\lambda}_2(\delta)/\mu) }}{\delta\hat{c}_1+(1-\delta)c_H} \bigg\}\bigg]. 
\end{align}
}
\end{table}
To avoid large social cost, we then take the smallest feasible solutions in (50) as
\begin{table}[H]
\small{
\begin{align}     
 \tilde{\lambda}_1(\delta)(\tilde{\lambda}_2(\delta)) =\frac{M_k' - \sqrt{M_k^{'2} -(\delta\hat{c}_1+(1-\delta)c_j)(1+ \sum\limits_{j=2}^N\tilde{\lambda}_j(\delta) /\mu) }}{\delta\hat{c}_1+(1-\delta)c_j}, 
        \end{align}
\begin{align}
  \tilde{\lambda}_i(\delta)(\tilde{\lambda}_1(\delta)) = \frac{M_i - \sqrt{M_i^{2} - c_i(1+ \sum\limits_{j \ne i, 1}^N\tilde{\lambda}_j(\delta)/\mu)}}{c_i}, \;\; i = 2, 3, \cdots, N,
\end{align}}
\end{table}
where 
\begin{table}[H]
\small{
\begin{align}
M_k' = \delta\sqrt{1+\sum\limits_{j=2}^N\lambda_j^*/\mu}(p_H\sqrt{c_H} + (1-p_H)\sqrt{c_k})+(1-\delta)\sqrt{(1+ \sum\limits_{j\ne i}^N\tilde{\lambda}_j(\delta))c_k }, \;\; k \in \{H, L\},  \nonumber
\end{align}
}
\end{table}
which is the same as the equations in Proposition \Romannum{4}.7. We want to use induction method to show (51)-(52) have unique solutions. When $N = 2$, (51)-(52) are equivalent to
\begin{align}     
 \tilde{\lambda}_1(\delta) =\max\bigg\{&\frac{M_L' - \sqrt{M_L^{'2} -(\delta\hat{c}_1+(1-\delta)c_L)(1+ \tilde{\lambda}_2(\delta) /\mu) }}{\delta\hat{c}_1+(1-\delta)c_L},\nonumber \\
 &\frac{M_H' - \sqrt{M_H^{'2} -(\delta\hat{c}_1+(1-\delta)c_H)(1+ \tilde{\lambda}_2(\delta) /\mu) }}{\delta\hat{c}_1+(1-\delta)c_H}
 \bigg\},
        \end{align}
\begin{align}
   \tilde{\lambda}_2(\delta) = \frac{M_2 - \sqrt{M_2^{2} - c_2(1+ \tilde{\lambda}_1(\delta)/\mu)}}{c_2}.
\end{align}
We rewrite (53)-(54) as
\begin{align}     
 \tilde{\lambda}_1(\delta)(\tilde{\lambda}_2(\delta)) =\frac{M_j' - \sqrt{M_j^{'2} -(\delta\hat{c}_1+(1-\delta)c_j)(1+ \tilde{\lambda}_2(\delta) /\mu) }}{\delta\hat{c}_1+(1-\delta)c_j}, 
        \end{align}
\begin{align}
   \tilde{\lambda}_2(\delta)(\tilde{\lambda}_1(\delta)) = \frac{M_2 - \sqrt{M_2^{2} - c_2(1+ \tilde{\lambda}_1(\delta)/\mu)}}{c_2},
\end{align}
where
\begin{align}
    M_j' = \delta\sqrt{1+\lambda_2^*/\mu}(p_H\sqrt{c_H} + (1-p_H)\sqrt{c_L})+(1-\delta)\sqrt{(1+ \tilde{\lambda}_2(\delta))c_j }. \nonumber
\end{align}
To show (55)-(56) have unique solutions in the range of $\hat{\lambda}_1<\tilde{\lambda}_1(\delta)<\lambda_1^{*}(c_j), \hat{\lambda}_2<\tilde{\lambda}_2(\delta)<\lambda_2^{*}$, denote 
\begin{align}
    g(\tilde{\lambda}_1(\delta)) = \tilde{\lambda}_2(\delta)(\tilde{\lambda}_1(\delta)) - \tilde{\lambda}_2(\delta)_1(\tilde{\lambda}_1(\delta)), \nonumber
\end{align}
 where $ \tilde{\lambda}_2(\delta)_1(\tilde{\lambda}_1(\delta))$ is the inverse function to $\tilde{\lambda}_1(\delta)(\tilde{\lambda}_2(\delta))$ in (55) with variable $\tilde{\lambda}_1(\delta)$.
By taking first-order and second-order derivatives of $\tilde{\lambda}_1(\delta)(\tilde{\lambda}_2(\delta))$ in (55) and of $\tilde{\lambda}_2(\delta)(\tilde{\lambda}_1(\delta))$ in (56), we can find that $\tilde{\lambda}_1(\delta)(\tilde{\lambda}_2(\delta))$ is convex and strictly increasing in $\tilde{\lambda}_2(\delta)$, and  $\tilde{\lambda}_2(\delta)(\tilde{\lambda}_1(\delta))$ is convex and strictly increasing in $\tilde{\lambda}_1(\delta)$. Thus $g(\tilde{\lambda}_1(\delta))$ is convex in $\tilde{\lambda}_1(\delta)$. Additionally, we have
\begin{align}
    g(\hat{\lambda}_1) &> 0, \nonumber \\
    g(\lambda_1^{*}(c_j)) &< 0, \nonumber \\
    g'(\lambda_1^*(c_j)) &> 0. \nonumber
\end{align}
Thus there exists unique $\tilde{\lambda}_1^0(\delta)$ in $\hat{\lambda}_1 < \tilde{\lambda}_1(\delta) < \lambda_1^*(c_j)$ satisfying $g(\tilde{\lambda}_1^0(\delta)) = 0$. We plot $g(\tilde{\lambda}_1(\delta))$ in Figure 9. Then 
there exist unique solutions to (53) and (54) in the feasible range of $\hat{\lambda}_1 < \tilde{\lambda}_1(\delta) < \lambda_1^*(c_j)$ and $\hat{\lambda}_2 < \tilde{\lambda}_2(\delta) < \lambda_2^*$.
\begin{figure}
\centering
 \includegraphics[height=2.5in, width=3in]{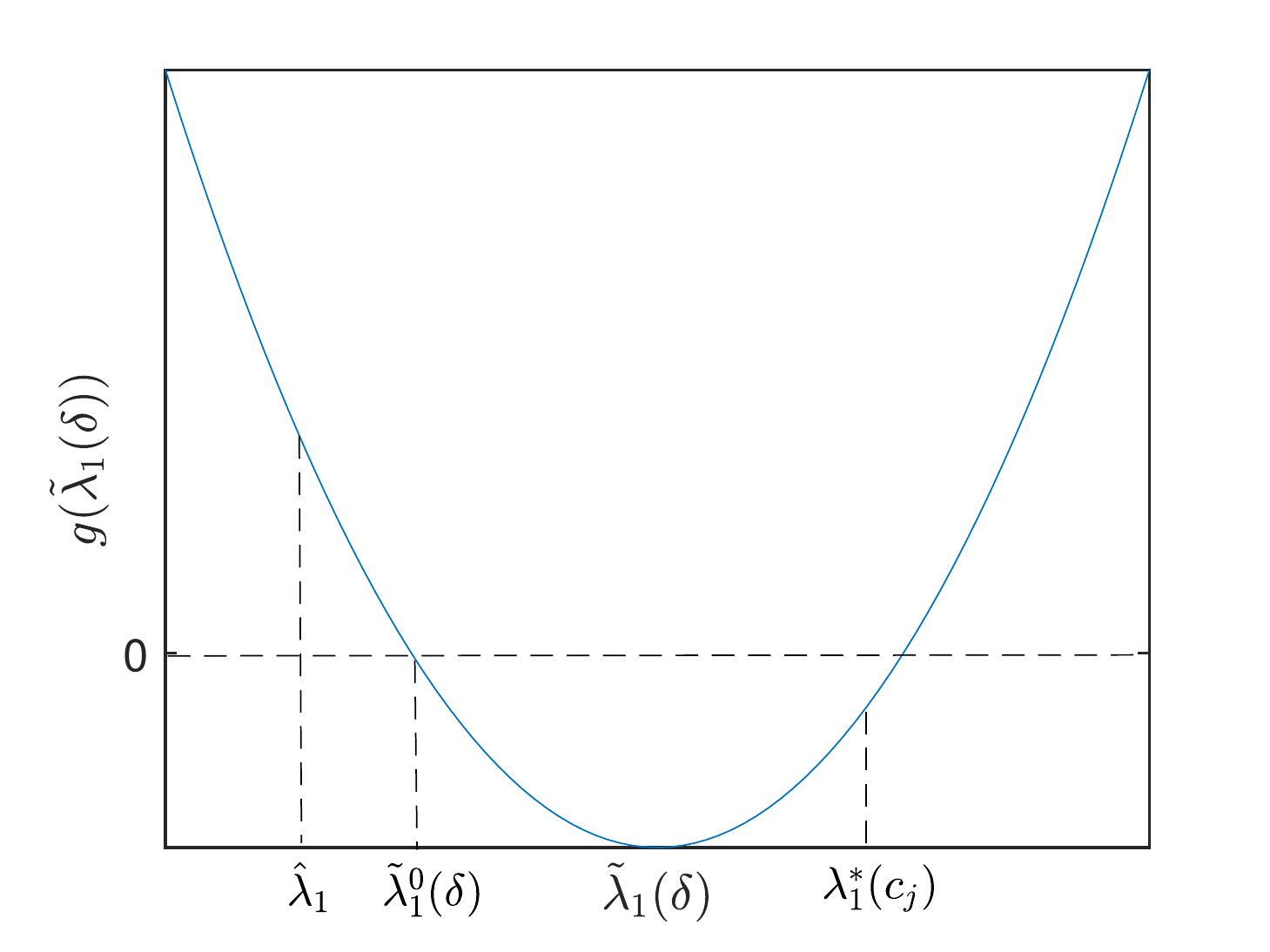}
\caption{$ g(\tilde{\lambda}_1(\delta))$ versus $\tilde{\lambda}_1(\delta)$ in Appendix I.}
 \end{figure}

 Suppose that when $N = M-1$, the equations in Proposition \Romannum{4}.7 have unique solutions. With induction method, we need to prove when $N = M$, the equations in Proposition \Romannum{4}.7 have unique solutions. Similar to (55)-(56), we can rewrite $\tilde{\lambda}_i(\delta)$ as a function of $\tilde{\lambda}_j(\delta)$ and the $\tilde{\lambda}_i(\delta)$ is convex and strictly increasing in each $\tilde{\lambda}_j(\delta)$, where $i\in\{1, \cdots, M\}$ and $j \ne i$. If we introduce $\tilde{\lambda}_M(\delta)$ as in (55)-(56) into other $\tilde{\lambda}_i(\delta)$ as in (55)-(56), where $i\in\{1, \cdots, M-1\}$, we have $\tilde{\lambda}_i(\delta)$ is still convex and strictly increasing in $\tilde{\lambda}_j(\delta)$, where $j \in \{1, \cdots, M-1\}$ and $j \ne i$. Since we know when $N = M-1$, the equations in Proposition \Romannum{4}.7 have unique solutions. Then after introducing $\lambda_M$ as in (55)-(56) into other $\lambda_i$ as in (55)-(56), the new $M-1$ equations also have unique solutions. Then we prove that when $N = M$, the equations in Proposition \Romannum{4}.7 have unique solutions.



\ifCLASSOPTIONcaptionsoff
  \newpage
\fi

\end{document}